\documentclass[aps,pra,a4paper,twocolumn,superscriptaddress,floatfix]{revtex4}

\usepackage[latin1]{inputenc}
\usepackage{amsthm}
\usepackage{amssymb}
\usepackage{graphicx}
\usepackage{amsmath}
\usepackage{bbold}
\usepackage{bbm}
\usepackage{mathtools}
\newcommand{\Tr}{{\mathrm{Tr}}}
\usepackage{nonfloat}
\usepackage{color}
\usepackage[pdftex]{hyperref}
\usepackage{braket}
\usepackage{graphicx}
\usepackage{dsfont}

\DeclareGraphicsExtensions{.jpeg, .png, .pdf}

\newtheorem{Def}{Definition}
\newtheorem{thm}{Theorem}
\newtheorem{prop}[thm]{Proposition}
\newtheorem{lemma}[thm]{Lemma}
\newtheorem{cor}[thm]{Corollary}
\newtheorem{ex}{Example}
\newtheorem{cj}{Conjecture}

\begin{document}

\title{Entanglement--Breaking Indices}
\author{L. Lami}
\affiliation{Scuola Normale Superiore, I-56126 Pisa, Italy.}

\author{V. Giovannetti}
\affiliation{NEST, Scuola Normale Superiore and Istituto Nanoscienze-CNR, I-56127 Pisa, Italy.}

\begin{abstract}
We study a set of new functionals (called entanglement--breaking indices) which characterize how many local iterations of a given (local) quantum channel are needed in order to completely destroy the entanglement between the system of interest over which the transformation is defined and an external ancilla. 
The possibility of contrasting  the noisy effects introduced by the channel iterations via the action of intermediate ({\it filtering}) transformations is analyzed. We provide some examples in which our functionals can be exactly calculated.
The differences between unitary and non-unitary filtering operations are analyzed showing that, at least for systems of dimension $d$ larger than or equal to 3, the non-unitary choice is preferable (the gap between the performances of the two cases being divergent in some cases). For $d=2$ (qubit case) on the contrary no evidences of the presence of such gap is revealed: we conjecture that for this special case unitary filtering transformations are optimal.
The scenario in which more general filtering protocols are allowed is also discussed in some detail. The case of a depolarizing noise acting on a two--qubit system is exactly solved in a general case.
\end{abstract}

\maketitle

\section{Introduction} \label{sec intro}

The theory of quantum channels~\cite{HOLEVOBOOK,WolfQC} is one of the cornerstones of the growing framework of quantum information~\cite{BennettShor}. The reason for this centrality is that a quantum channel is the natural description of the general dynamics of an open quantum system. That is, every external noise acting on a quantum system can be modeled through a suitable quantum channel, which transforms states of a quantum system in other states of the same system. From a mathematical point of view, a quantum channel can be thought of as a unitary interaction with an external environment which is  discarded  (Stinespring representation), as an intrinsic operation involving only operators acting on our system (Kraus representation), or as an abstract linear, completely positive, trace--preserving superoperator (axiomatic approach).

While on  one hand the theory of quantum channels is the fundamental paradigm to model the noise acting on quantum systems, on the other hand one of the most important resources that can be stored (and subsequently deteriorated) in the same systems is the quantum entanglement~\cite{ENTAN}. The power of this genuinely new effect, which has no counterpart in the classical world, is one of the main reasons for the interest of the scientific community in quantum information. Therefore, it is important to base a classification system of the noise introduced by a quantum channel only on its local action on the entanglement of a bipartite system.

Following the guidelines of Ref.~\cite{V}, we consider here those system whose noise can be thought of as a single elementary process which is iterated step by step, as the (discretized) time goes on. A fundamental assumption, which is generally well-founded from the experimental point of view, is that the various elementary steps are completely independent with each other. With this hypothesis, the action of the noise becomes a stroboscopic Markov process and can be modeled by the $n$--fold iteration of a given quantum channel. The main goal of this paper is to develop some functionals (called \emph{entanglement--breaking indices}) which characterize \emph{how many local iterations of a given channel are needed in order to completely destroy the global entanglement}. In particular, we will consider the possibility of improving the performance of the system via the action of  {\it filtering} operations~\cite{V}, i.e. quantum channels which are introduced between two consecutive iterations of the noise with the purpose of {\it protecting} the entanglement in the system. 
Our endeavor is somehow related  to the idea of quantum subdivision capacities  recently introduced  by M\"{u}ller-Hermes,  Reeb, and Wolf  in Ref.~\cite{MULLER}. These Authors considered the possibility of improving the quantum capacity~\cite{HOLEVOBOOK} associated with the time evolution of an assigned  dynamical semigroup, by interfering with the induced noise via the action of intermediate coding/decoding operations which are applied {\it while} the noise is still tampering the system, not just {\it before/after} it has already affected the communication. 
As a matter of fact, our filtering operations can  be seen as special instances of such intermediate operations. 
Differently from~\cite{MULLER} however, in this paper we focus on single channel uses scenarios where strategies that involve parallel encodings over multiple channel uses are not permitted. Furthermore, while Ref.~\cite{MULLER} deals with maps which are infinitely divisible and focuses on the continuous time-evolution limit,  our approach applies to iterations of arbitrary (not necessarily divisible) channels which operate in a stroboscopic fashion. 
Last but not the least, the figure of merit  we analyze (i.e. the entanglement--breaking property of a given concatenation of transformations) is stronger than simply requiring that the associated 
quantum capacity is null. 

Examples of quantum channels that allow for filtering transformations which protect the breaking of entanglement induced by the noise,  have been provided in Refs.~\cite{V,Vnonmaxent} and classified under the name  of {\it amendable} channels. In the present work we clarify several technical aspects associated with the filtering process showing that it is strongly influenced by the surrounding assumptions. In particular we prove that unitary filtering is in general not optimal at least when the dimension $d$  of the system of interest is larger than or equal to 3. For $d=2$ (qubit case) however this seems not be the case and we conjecture  that for this special case unitary filtering transformations provide the best protection under the iterations of the noise.
\\

In what follows greek letters such as $\phi$ or $\psi$ will typically denote quantum channels; the operation of composition, simply denoted by the juxtaposition $\phi\psi$, indicates the channel resulting from the consecutive application of $\psi$ firstly, and of $\phi$ secondly, while we shall use the symbol $\phi^n$ to represent the $n$-fold concatenation of a given channel $\phi$. The notation $\mathbf{CPt}_d$ stands for the set of linear, completely positive, trace--preserving maps acting on states of a $d$--dimensional system. Instead, the restricted set of unitary operations will be denoted by $\mathbf{U}_d$.

In Sec.~\ref{sec EBi} we  start by formalizing the notions of entanglement--breaking indices and of filtering operations, and discuss some basics properties.
In Sec.~\ref{sec:ex} we present few examples of channels for which explicit expressions for the indices can be obtained. 
Sec.~\ref{sec U vs nU filt} and  Sec.~\ref{sect filt q} instead deal with the difference between unitary and non-unitary filtering operations. In particular in Sec.~\ref{sec U vs nU filt} we show that for systems of dimension $d$ greater than or equal to 3, there are cases in which non-unitary filters perform much better than the unitary ones. 
In  Sec.~\ref{sect filt q} we conjecture that this should not be the case for systems of dimension 2 (i.e. a qubit), providing evidences and some preliminary results.
Finally, in Sec.~\ref{SEC:LOCC} we address the case where arbitrary Local Operations and Classical Communication (LOCC), or even separable~\cite{BENNETT,ENTAN} protocols, are allowed to be used as  filtering transformations. 
Conclusions and remarks are presented in Sec.~\ref{sec conc}.

\section{Entanglement -- Breaking Indices} \label{sec EBi}

In this section we introduce the functionals to be studied through the rest of the paper. Our setup is as follows: Alice and Bob share an entangled state, but Alice's half of the global system is repeatedly affected by some noise represented by the quantum channel $\phi$. Alice and Bob are supposed to be not able to communicate with each other (neither with classical nor with quantum means). Of course, this is a significant restriction, and other more sophisticated scenarios could be considered. For example, we could allow Alice and Bob to communicate with a classical device. We will discuss some nontrivial facts about this framework in the Appendix.

Our first concern is to recall some basic facts about the entanglement--breaking channels, being these a fundamental tool in our approach.

\subsection{Entanglement--Breaking Channels} \label{subsec EB}
A particularly noisy class of quantum channel is composed of those transformation that always produce a global separable state when applied locally to a generic input state. These maps are called \emph{entanglement--breaking}, and their set will be denoted by $\mathbf{EBt}$ (possibly with the subscript $d$ if we want to specify the dimension of the system on which we are acting). It is worth noting that $\mathbf{EBt}$, just like $\mathbf{CPt}$, is a compact convex set. Moreover, it turns out from the very definition that
\begin{equation} \phi\in\mathbf{EBt}\,,\ \psi\in\mathbf{CPt}\quad\Rightarrow\quad \phi\psi,\,\psi\phi\in\mathbf{EBt}\ . \label{EB propagates}\end{equation}

The fundamental characterization theorem concerning the entanglement--breaking channels is proved in~\cite{HorodeckiShorRuskai}. It states that the following facts are equivalent:
\begin{itemize}
\item the channel $\phi$ is entanglement--breaking;
\item the \emph{Choi state} $(\phi\otimes I)(\Ket{\varepsilon}\!\!\Bra{\varepsilon})$, with $\Ket{\varepsilon}$ maximally entangled state, is \emph{separable} (i.e. it can be written as a convex combination of product states);
\item and there exists an operative \emph{measurement + re--preparation} interpretation of the form
\begin{equation} \phi(X)\ =\ \sum_i \rho_i\, \text{Tr}\, [E_i X]\ , \label{Holevo form} \end{equation}
where the $\{ \rho_i \}$ are density matrices, and the positive operators $\{ E_i \}$ satisfy the sum rule $\sum_i E_i=\mathds{1}$. This expression is called \emph{Holevo form}~\cite{Holevoform}.
\end{itemize}

Let us take a close look to the qubit case. A generic state of a two--dimensional system can be written in the \emph{Bloch sphere} representation as
\begin{equation} \rho=\frac{\mathds{1}+\vec{r}\cdot\vec{\sigma}}{2}\quad , \label{Bloch repr q} \end{equation}
where $\vec{\sigma}=(X,Y,Z)$ denotes the vector of Pauli matrices, and $|\vec{r}|\leq 1$. The pure states are exactly those states $\rho$ whose associated vector $\vec{r}$ has unit modulus. Now, the action of a quantum channel $\phi$ is completely specified by the $3\times 3$ real matrix $M$ and the $3$--dimensional vector $c$ such that
\begin{equation} \phi\left(\frac{\mathds{1}+\vec{r}\cdot\vec{\sigma}}{2}\right) \ = \ \frac{\mathds{1}+(M\vec{r}+\vec{c})\cdot\vec{\sigma}}{2}\quad . \label{Bloch repr} \end{equation}
Consequently, we will sometimes indicate the channel $\phi$ with the notation $(M,c)$. If $\Ket{\varepsilon}=(\Ket{00}+\Ket{11})/\sqrt{2}$ is the maximally entangled state of two qubits, we can write the remarkable equality
\begin{equation} 4 \Ket{\varepsilon}\!\!\Bra{\varepsilon}\ =\ \mathds{1} + \sum_{i=1}^3 \sigma_i \otimes \sigma_i^T\ =\ \sum_{\mu = 0}^3  \sigma_\mu \otimes \sigma_\mu^T\ . \label{max ent Pauli} \end{equation}
Thanks to~\eqref{max ent Pauli}, the Choi state associated to $(M,c)$ takes the form
\begin{equation} R_{(M,c)}\ =\ \frac{1}{4}\, \left(\, \mathds{1} + (\vec{c}\cdot\vec{\sigma})\otimes \mathds{1}\, +\, \sum_{i,j=1}^3 M_{ij} \, \sigma_i \otimes \sigma_j^T\, \right)\ . \label{Fano} \end{equation}
If the channel is unital (that is, $c=0$), the entanglement--breaking condition (i.e. the separability condition for $R_{(M,0)}$) becomes extremely simple~\cite{AlgoetFujiwara,RuskaiEBqubit}:
\begin{equation} (M,0)\in\mathbf{EBt}_2\quad \Leftrightarrow\quad \|M\|_1\leq 1\ . \label{EB q u} \end{equation}
Here we used the standard notation
\begin{equation} \|A\|_p\ =\ \left(\,\text{Tr}\left[(A^\dag A)^{p/2}\right]\,\right)^{1/p} \label{Schatten} \end{equation}
for the Schatten norm of index $1\leq p\leq \infty$.

\subsection{Entanglement--Breaking Indices} \label{subsec EBi}
As previously stated, the main goal of this paper is to classify the amount of noise introduced by a quantum channel only by means of the effect of its local iterations on a global bipartite entanglement. The first step in our approach is the identification of some interesting functionals (which we call \emph{indices} because they are integer--valued) defined on the set of quantum channels. We postpone our comments after the mathematical definition.

\begin{Def}[Entanglement--Breaking Indices] \label{EBi} $\\$
Let $\phi \in \mathbf{CPt}$ be a quantum channel. Define
\begin{equation} n(\phi) = \min{\{n\geq 1:\ \phi^n \in \mathbf{EBt}\}}\ , \label{n}  \end{equation}
\vspace{-4ex}
\begin{multline} \mathcal{N}_U(\phi)\ =\ \min{\{\,n\geq 1:\ \forall \ \mathcal{U}_1,\ldots,\mathcal{U}_{n-1} \in \mathbf{U},}\\
\phi\,\mathcal{U}_1\phi\ldots\phi\,\mathcal{U}_{n-1}\phi \in \mathbf{EBt}\, \}\ , \label{NU}  \end{multline}
\vspace{-4ex}
\begin{multline} \mathcal{N}(\phi)\ =\ \min{\ \{\,n\geq 1:\ \forall \ \psi_1,\ldots,\psi_{n-1} \in \mathbf{CPt},}\\
\phi\psi_1\phi\ldots\phi\psi_{n-1}\phi \in \mathbf{EBt}\, \}\ . \label{N} \end{multline}

\vspace{1ex}
For an entanglement--breaking (EB) channel all these indices are set equal to $1$ by definition. Moreover, it is implicitly understood that the minimum of an empty set and the maximum of an unlimited set should be posed equal to $+\infty$, which becomes in this way a legitimate value of the functionals defined. We call \emph{filters} the maps used between repeated applications of a channel to reduce its entanglement--breaking properties (the $\mathcal{U}$'s of~\eqref{NU} or the $\psi$'s of~\eqref{N}). Given a subset of filters $F\subseteq\mathbf{CPt}$, one can consider more generally the restricted filtered index:
\begin{multline} \mathcal{N}_F(\phi) \ =\ \min{\{\,n\geq 1:\ \forall \ \psi_1,\ldots,\psi_{n-1} \in F,}\\
\phi\psi_1\phi\ldots\phi\psi_{n-1}\phi \in \mathbf{EBt}\, \}\ . \label{N restr} \end{multline}
Obviously, equation~\eqref{N restr} reduces itself to~\eqref{NU} if $F=\mathbf{U}$, and to~\eqref{N} if $F=\mathbf{CPt}$, respectively.
\end{Def}

Several observations and explanations are necessary. These functionals represent an inverse measure of the noise introduced in the system by a given channel. The smaller is the value of the index, the more dangerous for the entanglement is the action of the channel. For example, all these indices assume the value $1$ for entanglement--breaking channels and $+\infty$ for the unitary transformations. 

Firstly, let us discuss the \emph{direct $n$--index} defined by~\eqref{n}, since it is the most intuitive one. It is nothing but the smallest number of direct, serial applications of a given channel such that the global transformation becomes entanglement--breaking. In this situation Alice plays no role against the noise. Her subsystem simply suffers it a few at a time, and there is no possibility to contrast or delay its action. This quantity already appears in~\cite{V}, though it is indicated by $n_c$ there; we adopt the shorthand $n$. 

The other functionals are \emph{filtered indices}. This means that Alice chooses to \emph{play an active role against the noise} affecting her subsystem. Her strategy is the simplest possible, consisting of the application of some filters between an action of the noisy channel and the subsequent one. A filter is nothing but a (local) quantum channel that is chosen by Alice in such a way as to preserve the entanglement with Bob as best as she can. In this context, there are mainly two possible scenarios.

\begin{itemize}
\item In~\eqref{NU}, we consider only unitary filtering maps $\mathcal{U}_i$ (allowing them to be changed from time to time).
\item In~\eqref{N} we optimize over all the possible sets of $\mathbf{CPt}$ operations implemented by Alice. In other words, we admit the possibility that non--unitary filters $\psi_i$ are used.
\end{itemize}

We stress here that the one described above is only the simplest among a rich variety of possible scenarios. Many others possibilities can be equally interesting from an experimental point of view. For a more detailed discussion about the directions in which this simple framework can be generalized, we refer the interested reader to Section~\ref{SEC:LOCC}.

\subsection{Elementary Properties}
Our first concern is the analysis of the elementary properties of these entanglement--breaking indices. Their proofs (which we omit for the sake of brevity) are directly related to the operational meaning of our functionals, as outlined in the previous section. Let us group all together in a proposition: \\

\begin{prop}[Elementary Properties] $\\$
Let $\phi\in\mathbf{CPt}$ be a quantum channel. Then the following properties hold.
\begin{description}

\item[Unitary conjugation:] If $\,\mathcal{U},\mathcal{V}\in\mathbf{U}$ are unitary evolutions, then
\begin{gather}
n\,(\mathcal{U}\phi\,\mathcal{U}^\dag) \equiv n(\phi)\ , \label{nUC}\\
\mathcal{N}_U(\mathcal{U}\phi\mathcal{V})\equiv\, \mathcal{N}_U(\phi)\ ,\quad \mathcal{N}(\mathcal{U}\phi\mathcal{V}) \equiv\, \mathcal{N}(\phi)\ . \label{NUC}
\end{gather}

\item[Composition with generic channels:] Let $\psi\in\mathbf{CPt}$ be another quantum channel. Then
\begin{equation} \mathcal{N}(\phi\psi)\ \leq\ \mathcal{N}(\phi),\ \mathcal{N}(\psi)\ . \end{equation}
Here the commas denote alternative options.

\item[Elementary inequalities:] The following elementary inequalities hold:
\begin{equation} n(\phi)\ \leq\ \mathcal{N}_U(\phi)\ \leq\ \mathcal{N}(\phi)\ . \label{elem ineq} \end{equation}
Examples of maps which exhibit a finite gap between $n(\phi)$ and $\mathcal{N}_U(\phi)$ were first given in Refs.~\cite{V,Vnonmaxent} (these maps were called {\it amendable}).

\item[Reduction to the extreme points:] Denote by $\,\mathcal{C}(F)$ the convex hull of a certain set of filters $F\subseteq\mathbf{CPt}$. Consider the extreme points $e\mathcal{C}(F)$ of the convex set obtained. Then
\begin{equation} \mathcal{N}_F(\phi) \equiv \mathcal{N}_{e\mathcal{C}(F)}(\phi)\ . \end{equation}

\end{description}
\end{prop}

Now, let us analyze some less trivial properties of our indices. Recall that every quantum channel $\psi$ (in particular, the filters involved in~\eqref{N}) admits a Stinespring representation. In other words, $\psi$ can be seen as the (non-unitary) restriction of a global unitary evolution in a greater system. We can exploit this physical property in order to reduce the set of filters to only the unitary ones. However, this is done at the price of expanding the dimension of the system. In the following, suppose that our system has dimension $d$. Consider another ``environment'' $E$ of dimension $d^2$, and denote by $\Ket{0}\in\mathcal{H}_E$ a pure state of $E$. The associated \emph{completely depolarizing channel} $D_0\in\mathbf{EBt}_{d^2}$ acts by definition as
\begin{equation} D_0 (X)\equiv\Ket{0}\!\!\Bra{0}\ \Tr{X}\ . \label{D0 env}\end{equation}

With these preliminary discussion, we can prove the following theorem. \\

\begin{thm}[Stinespring Dilation of Filtered Indices] \label{St dil} $\\$
Let $\phi\in\mathbf{CPt}_d$ be a quantum channel. With the notation of~\eqref{D0 env}, one has
\begin{equation} \mathcal{N}\,(\phi)\ =\ \mathcal{N}_U\,(\phi\otimes D_0)\ , \label{St dil N}\end{equation} 
\end{thm}

\begin{proof}
Consider a filtering strategy $\phi\psi_1\phi\ldots\phi\psi_{n-1}\phi$ implemented by Alice. Take the global unitary evolutions $\mathcal{U}_i\in\mathbf{U}_{d^3}$ (acting as $\mathcal{U}_i(X)=U_i X U_i^\dag$) which represent the filters $\psi_i$ in Stinespring form:
\begin{equation*} \psi_i(X)\ =\ \text{Tr}_E\ [\ U_i\ X\otimes\Ket{0}\!\!\Bra{0}\ U_i^\dag\ ]\ \ . \end{equation*}
In the previous equation the first degree of freedom corresponds to our system, while the second one is the (fictitious) environment. We will maintain this notation in what follows. As can be easily seen, for each $n\geq 1$ we can write
\begin{multline} \phi\psi_1\phi\ldots\phi\psi_{n-1}\phi\ \otimes\ D_0\ =\\
=\ (\phi\otimes D_0)\ \mathcal{U}_1\, (\phi\otimes D_0)\,\ldots\,(\phi\otimes D_0)\ \mathcal{U}_{n-1}\, (\phi\otimes D_0)\ . \label{St dil filters} \end{multline}
Indeed, consider for example the case $n=2$ :
\begin{gather*} (\phi\otimes D_0)\,\, \mathcal{U}\, (\phi\otimes D_0)\, (X)\ =\\
=\ (\phi\otimes D_0)\ \mathcal{U}\, \left(\,\phi \left(\text{Tr}_E X \right) \otimes \Ket{0}\!\!\Bra{0}\, \right)\ =\\
=\ \phi\, \left(\, \text{Tr}_E\, [\,U \left(\, \phi \left( \text{Tr}_E X \right) \otimes \Ket{0}\!\!\Bra{0}\, \right) U^\dag\, ]\, \right)\, \otimes\, \Ket{0}\!\!\Bra{0}\ =\\
=\ \phi\, \left(\, \psi \left( \phi \left( \text{Tr}_E X \right)\right)\, \right) \otimes \Ket{0}\!\!\Bra{0}\ =\ \left(\phi\psi\phi\otimes D_0\right)\ (X)\ . \end{gather*}
Moreover, it is worth noting that to each unitary family $\{\,\mathcal{U}_i\,\}\subseteq\mathbf{U}_{d^3}$ we can associate a corresponding family $\{\psi_i\}\subseteq\mathbf{CPt}_d$ such that~\eqref{St dil filters} is satisfied. Since $D_0$ is a completely depolarizing channel (i.e. its images are all proportional to a fixed matrix), it can be immediately verified that for each $\eta\in\mathbf{CPt}_d$
\begin{equation*} \eta\otimes D_0\in\mathbf{EBt}_{d^3}\quad\Leftrightarrow\quad \eta\in\mathbf{EBt}_d\ . \end{equation*}
Therefore, we can directly prove~\eqref{St dil N}:
\begin{gather*} \mathcal{N}(\phi)\ \equiv\ \min{\ \left\{\,n\geq 1:\ \ \forall \ \psi_1,\ldots,\psi_{n-1} \in \mathbf{CPt}_d,\right.}\\
\left.\phi\psi_1\phi\ldots\phi\psi_{n-1}\phi \in \mathbf{EBt}_d\, \right\}\ =\\
=\ \min{\ \left\{\,n\geq 1:\ \ \forall \ \psi_1,\ldots,\psi_{n-1} \in \mathbf{CPt}_d,\right.}\\
\left.\phi\psi_1\phi\ldots\phi\psi_{n-1}\phi\, \otimes\, D_0 \in \mathbf{EBt}_d\, \right\}\ =\\
=\ \min{\ \left\{\,n\geq 1:\ \ \forall \ \mathcal{U}_1,\ldots,\mathcal{U}_{n-1} \in \mathbf{U}_{d^3}\, ,\right.}\\
(\phi\otimes D_0)\ \mathcal{U}_1\, (\phi\otimes D_0)\ldots(\phi\otimes D_0)\ \mathcal{U}_{n-1}\, (\phi\otimes D_0) \in\\
\left.\in \mathbf{EBt}_{d^3} \right\}\ \equiv\ \mathcal{N}_U(\phi\otimes D_0)\ . \end{gather*}
\end{proof}

\section{Examples}\label{sec:ex}
Through this section, we present a large variety of explicit, nontrivial examples of channels for which some entanglement--breaking indices can be calculated. This will help to explain the meaning of Definition~\ref{EBi}, and to become acquainted with it.

In what follows we will use extensively the Bloch sphere representation~\eqref{Bloch repr} of the qubit (i.e. $d=2$) channels. For unital qubit channels $\phi=(M,0)$, observe that~\eqref{EB q u} implies the simple equality
\begin{equation} n(\phi)\ =\ n(M,0)\ = \ \min{\,\lbrace\, n\geq 1 :\ \| M^n \|_1 \leq 1\, \rbrace}\ . \label{nfu2} \end{equation}

The first example of calculation of the direct $n$--index is presented in Ref.~\cite{V}. We report it here for the sake of completeness.

\begin{ex}[$n$--Index of Generalized Amplitude Damping Channels] \label{n GAD ex} $\\$
A fundamental physical process involving a system coupled to an environment in a thermal state is the spontaneous emission. In the case of a single qubit, this process can be described by a \emph{generalized amplitude damping} (GAD). The set of GADs is parametrized by the two real numbers $0\leq p\leq 1$ and $0\leq \gamma\leq 1$, linked to the time the interaction takes (or to its intensity) and to the temperature of the environment, respectively (see~\cite{NC}, p. 382). The action of a GAD on a given qubit state can be written as follows:
\begin{equation} GAD_{p,\gamma} \left(\begin{smallmatrix} a & b \\ b^* & c \end{smallmatrix}\right) = \left(\begin{smallmatrix} pa+\gamma (1-p) (a+c) & \sqrt{p}\, b \\ \sqrt{p}\, b^* & -pa+(1-(1-p)\gamma) (a+c) \end{smallmatrix}\right) . \label{GAD act}\end{equation}

As usual,~\eqref{Bloch repr} allows us to write the Bloch representation \mbox{$GAD_{p,\gamma}=\left( M_{p,\gamma},\, c_{p,\gamma} \right)$}, where
\begin{equation} M_{p,\gamma}\, =\, \left(\begin{smallmatrix} \sqrt{p} & 0 & 0 \\ 0 & \sqrt{p} & 0 \\ 0 & 0 & p \end{smallmatrix}\right)\ ,
\quad c_{p,\gamma}\, =\, (1-p)\,(2\gamma-1) \left(\begin{smallmatrix} 0 \\ 0 \\ 1 \end{smallmatrix}\right)\ . \label{GAD Bloch}\end{equation}

The composition rules of the GADs can be easily deduced for example by means of equation~\eqref{GAD Bloch}. It turns out that
\begin{gather} GAD_{p_1,\gamma_1}\ GAD_{p_2,\gamma_2}\ =\ GAD_{p_3,\gamma_3}\ ,\\
p_3\equiv p_1 p_2\ , \quad \gamma_3\,\equiv\,\frac{p_1(1-p_2)\gamma_2+(1-p_1)\gamma_1}{1-p_1 p_2}\ . \label{GAD comp gen} \end{gather}
In particular,
\begin{equation} GAD_{p,\gamma}^n\, \equiv\, GAD_{p^n,\gamma}\ . \label{GAD comp n} \end{equation}

Now, let us concern ourselves about the entanglement--breaking properties of the GADs. By applying the Positive Partial Transpose (PPT) condition~\cite{PeresPPT,HorodeckiPPT} to the Choi state associated to $GAD_{p,\gamma}$, the range of $p,\gamma$ which identifies an entanglement--breaking behavior can be easily deduced:
\begin{multline} GAD_{p,\gamma}\in\mathbf{EBt}_2\ \ \Longleftrightarrow\\
\Longleftrightarrow\ \ 0\, \leq\, p\, \leq\, f(\gamma)\,\equiv\, 1\,-\,\frac{2}{1+\sqrt{1+4\gamma(1-\gamma)}}\ . \label{GAD EB}\end{multline}

That said, we can easily calculate the direct $n$--index for the set of generalized amplitude damping channels. Indeed,~\eqref{GAD comp n} together with~\eqref{GAD EB} implies that
\begin{equation} n\,(GAD_{p,\gamma})\ =\ \Big\lceil\, \frac{\log f(\gamma)}{\log p} \,\Big\rceil\quad , \label{n GAD}\end{equation}
where the \emph{ceiling function} $\lceil\cdot\rceil$ is defined by
\begin{equation} \lceil x \rceil\,\equiv\,\min{\{ s\in\mathds{Z}\,:\ s\geq x \}}\ . \label{ceiling} \end{equation}
In~\eqref{n GAD}, we have supposed $p>0$; otherwise, we immediately know that $n(GAD_{0,\gamma})\equiv 1$. Moreover, observe that~\eqref{n GAD} returns $n=\infty$ as soon as $\gamma=1$ (with $p>0$). The GADs with $\gamma=1$ are often called simply \emph{amplitude damping}, and correspond to the modelization of a spontaneous emission interaction with an environment at zero temperature.

A pictorial representation of the regions of the space $p,\gamma$ identified by equation~\eqref{n GAD} can be found in \figurename~\ref{nGAD,fig} -- see also Ref.~\cite{V}.

\vspace{2ex}
\begin{figure}[ht] 
\centering
\includegraphics[height=8cm, width=8cm, keepaspectratio]{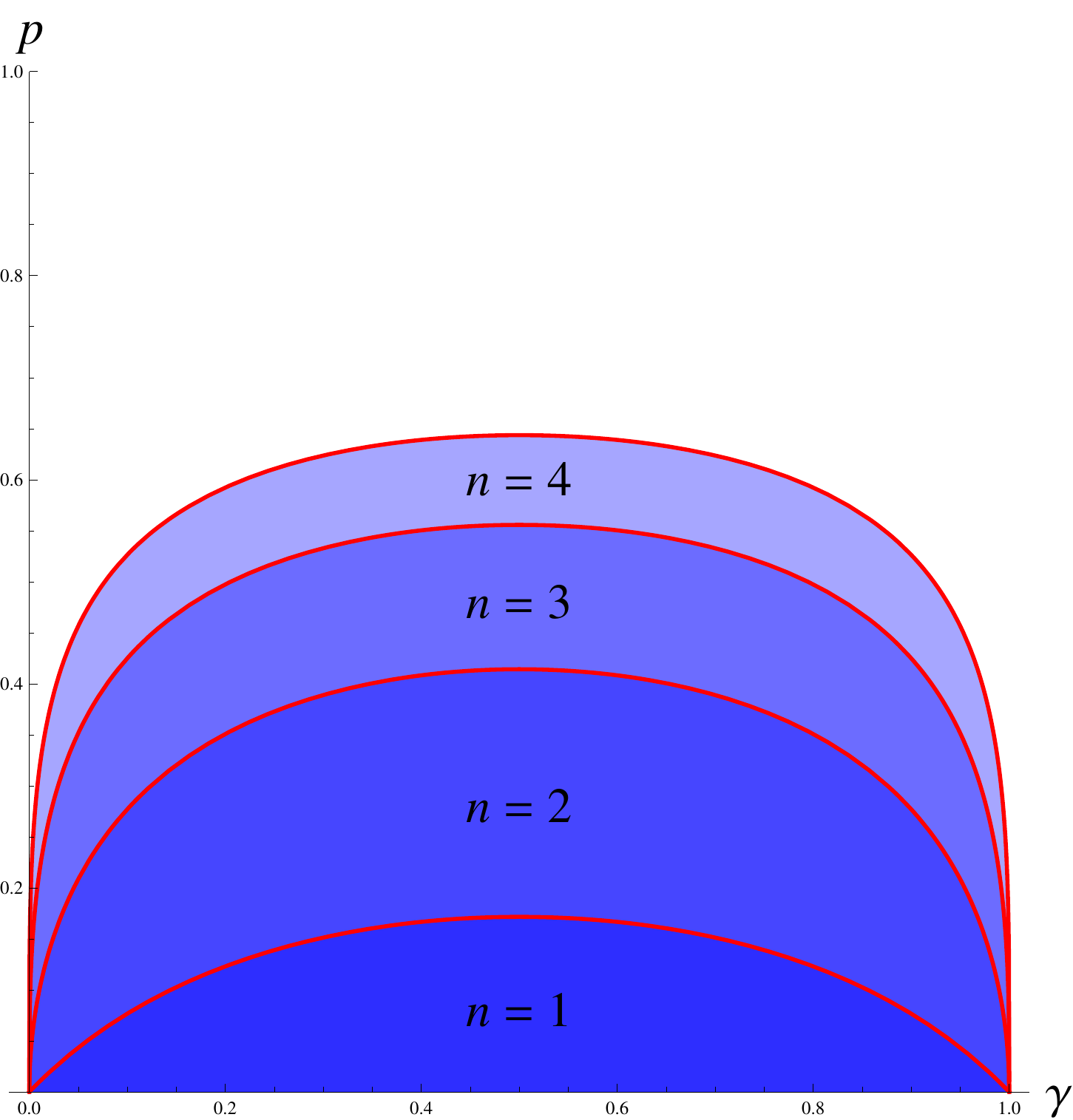}
\caption[$n$--Index for Generalized Amplitude Damping channels]{Graphic representation of the value of the direct $n$--index in the parameter space $\gamma,p$ of the GAD channels. The boundary points are always included in the adjacent region which has the lowest value of $n$.}
\label{nGAD,fig}
\end{figure}

\end{ex}

The previous example focused on the qubit case. However, there exists another famous class of channels acting \emph{in arbitrary dimension} for which the entanglement--breaking properties can be studied analytically. \\

\begin{ex}[$n$--Index of Depolarizing Channels] \label{n D ex} $\\$
The depolarizing channels are defined through a simple operative procedure on Alice's $d$--dimensional system. This procedure could be seen as the interaction with an environment, as usual, but it can be more simply visualized by involving a third human agent, named Eleonore. Eleonore takes Alice's state $\rho$ and secretly rolls a die. Depending on the outcome of the die, she gives back the system to Alice without performing any operation (with a certain probability $\lambda$), or she discards Alice's state and replaces it with the maximally mixed one~$\frac{\mathds{1}}{d}$ (with probability $1-\lambda$). In any case, Alice does not know the outcome of the die. Clearly, from her point of view, the state of the system transforms as follows:
\begin{equation*}  \rho\ \longmapsto\ \lambda\,\rho\, +\, (1-\lambda)\ \frac{\mathds{1}}{d}\ \ . \end{equation*}

The \emph{depolarizing channels} are thus defined by
\begin{equation} \Delta_\lambda\, \equiv\, \lambda I + (1-\lambda)\, \frac{\mathds{1}}{d}\, \text{\emph{Tr}}\ ,\quad -\frac{1}{d^2-1}\leq\lambda\leq 1\ , \label{D Ch}\end{equation}
where
\begin{equation} \frac{\mathds{1}}{d}\text{\emph{Tr}}\, :\, X\, \longmapsto\, \frac{\mathds{1}}{d}\, \text{\emph{Tr}} X\ . \label{dep id act} \end{equation}
It can be easily seen that the range of the parameter $\lambda$ in~\eqref{D Ch} is chosen in such a way as to guarantee that $\Delta_\lambda$ is always a completely positive (trace-preserving and unital) map. Observe that also a (little) range of negative values is allowed; this would not fit into our probabilistic operative definition, but this is going to be irrelevant. The laws of composition of the depolarizing channels are very simple:
\begin{equation} \Delta_{\lambda_1} \Delta_{\lambda_2} = \Delta_{\lambda_1 \lambda_2}\qquad \big(\Rightarrow\ \Delta_\lambda^n = \Delta_{\lambda^n}\big)\ . \label{D comp}\end{equation}

The class of depolarizing channels is important because its entanglement--breaking properties can be studied analytically. Indeed, in~\cite{HorodeckiDep} it is proved that
\begin{equation} \Delta_\lambda\in\mathbf{EBt}_d\quad\Longleftrightarrow\quad -\frac{1}{d^2-1}\leq\lambda\leq \frac{1}{d+1}\ . \label{D EB}\end{equation}

Thanks to~\eqref{D EB}, we can explicitly compute the actual value of the $n$--index for a depolarizing channel in arbitrary dimension. We are free to suppose $0<\lambda\leq 1$, since the values $\lambda\leq 0$ are immediately known to correspond to entanglement--breaking channels. Then we have
\begin{equation} n\,(\Delta_\lambda)\ =\ \Big\lceil\ \frac{\log\, (d+1)}{\log\,\frac{1}{\lambda}}\ \Big\rceil  \ . \label{n D} \end{equation}
For $\lambda=1$ (actually, $1^-$), this equation gives $n=\infty$, as expected (because $\Delta_1=I$). There are no other values of $\lambda$ sharing this property. The graphic of~\eqref{n D} is shown in \figurename~\ref{nD,fig}.

\vspace{3ex}
\begin{figure}[ht] 
\centering
\includegraphics[height=8cm, width=8cm, keepaspectratio]{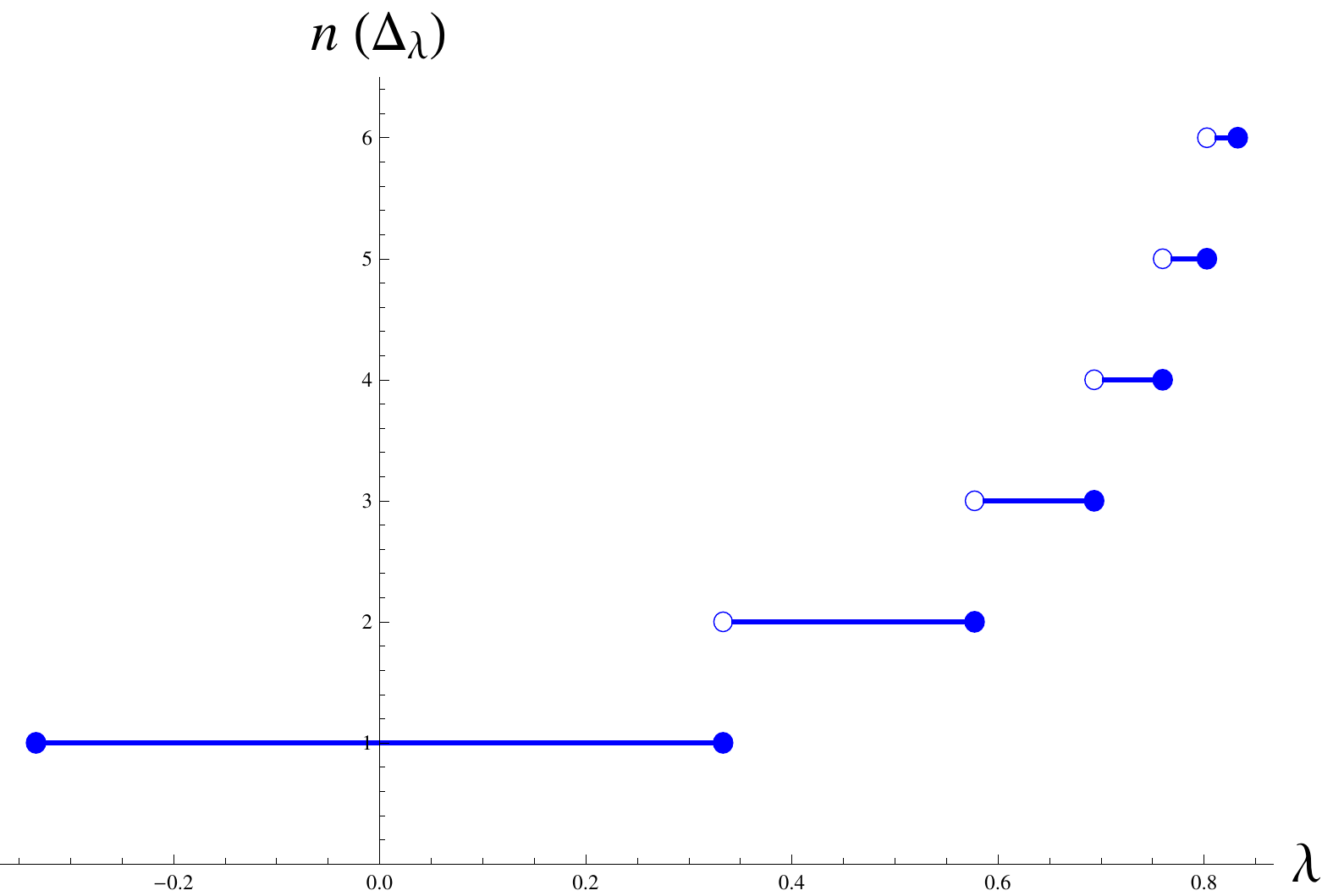}
\caption[$n$--Index for Depolarizing Channels ($d=2$)]{Graphic of the $n$--index as a function of the parameter $\lambda$ for a depolarizing channel. Here the qubit case $d=2$ is shown.}
\label{nD,fig}
\end{figure}

\end{ex}

Until this time, the discussion focused mainly on the $n$--index. Now, let us jump on the opposite side of the inequality~\eqref{elem ineq}. Because of the fact that a minimization over the entire set of $\mathbf{CPt}$ channels is required, the \mbox{$\mathcal{N}$--index} could seem a difficult functional to calculate in practice. Let us make an example to show that this is not always the case. \\

\begin{ex}[$\mathcal{N}$--Index of Depolarizing Channels] \label{N D ex} $\\$
In Example~\ref{n D ex} we introduced the important class of the depolarizing channels, acting on arbitrary $d$--dimensional systems (see~\eqref{D Ch}). We saw in~\eqref{n D} that their $n$--index can be explicitly computed. However, the question remains open, whether it is possible to enhance the entanglement preservation by means of the application of some filtering map. In other words, what can Alice do in order to preserve as much as possible the entanglement with Bob against the noisy action of Eleonore? The answer to this question is simple: \emph{she can do nothing}. This is the same as to say that all the entanglement--breaking indices are equal when calculated on a depolarizing channel:
\begin{equation} n\,(\Delta_\lambda) = \mathcal{N}_U(\Delta_\lambda) = \mathcal{N}\,(\Delta_\lambda)\, =\, \Big\lceil\, \frac{\log (d+1)}{\log\frac{1}{\lambda}}\, \Big\rceil  \ . \label{N D} \end{equation}
In what follows, we suppose as usual $\lambda>0$; otherwise, the depolarizing channels are already entanglement--breaking.

\begin{proof}[Proof of~\eqref{N D}]
Since~\eqref{n D} and~\eqref{elem ineq} hold, in order to prove~\eqref{N D} it suffices to show that
\begin{equation*} n(\Delta_\lambda)\, \geq\, \mathcal{N}\,(\Delta_\lambda)\ , \end{equation*}
i.e. that
\begin{multline*} \Delta_{\lambda^n}\in\mathbf{EBt}_d \ \Rightarrow\\
\Rightarrow\ \Delta_\lambda\psi_1 \Delta_\lambda\ldots \Delta_\lambda \psi_{n-1} \Delta_\lambda \in \mathbf{EBt}_d\ ,\\
\forall\ \ \psi_1,\ldots,\psi_{n-1}\in\mathbf{CPt}_d \ . \end{multline*}

Actually, the equality $n\,(\Delta_\lambda)=\mathcal{N}_U(\Delta_\lambda)$ can be seen as a direct consequence of the fact that \emph{the depolarizing channels commute with all the unitary evolutions}:
\begin{equation} \Delta_\lambda\,\mathcal{U}\equiv\,\mathcal{U}\,\Delta_\lambda\, ,\ \ \forall\ \mathcal{U}\in\mathbf{U}_d\, ,\ \ \forall\ -\frac{1}{d^2-1}\leq\lambda\leq 1\ . \label{D commute U}\end{equation}
Indeed, one could take~\eqref{D commute U} as the \emph{defining property} of the $\Delta_\lambda$s. However, the behavior of $\mathcal{N}(\Delta_\lambda)$ is a priori not obvious.

With the same notation as in~\eqref{dep id act}, it can be easily proved by induction that
\begin{multline} \Delta_\lambda\psi_1 \Delta_\lambda\ldots \Delta_\lambda \psi_{n-1} \Delta_\lambda\ =\ \lambda^n \psi_1\ldots\psi_{n-1}\ +\\
+\ (1-\lambda)\ \sum_{i=1}^{n-1}\, \lambda^i\, \left( \psi_1\ldots\psi_i \right)\left( \frac{\mathds{1}}{d} \right)\ \text{Tr}\ +\ (1-\lambda)\ \frac{\mathds{1}}{d}\ \text{Tr} \ . \label{D eq1}\end{multline}

Moreover, since $\Delta_{\lambda^n}$ is entanglement--breaking, and~\eqref{EB propagates} holds, we must have for every $1\leq i\leq n-1$
\begin{multline} \psi_1\ldots\psi_i\,\Delta_{\lambda^n}\,\psi_{i+1}\ldots\psi_{n-1}\ =\\
=\ \lambda^n \psi_1\ldots\psi_{n-1}\, +\, (1-\lambda^n)\, (\psi_1\ldots\psi_i)\left(\frac{\mathds{1}}{d}\right)\, \text{Tr}\, \in\, \mathbf{EBt}_d \, . \label{D eq2}\end{multline}
The generalization of~\eqref{D eq2} for the ``degenerate case'' $i=0$ can be immediately written as
\begin{multline} \Delta_{\lambda^n}\,\psi_{1}\ldots\psi_{n-1}\ =\\
=\ \lambda^n \psi_1\ldots\psi_{n-1}\ +\ (1-\lambda^n)\ \frac{\mathds{1}}{d}\ \text{Tr}\ \in\ \mathbf{EBt}_d \ . \label{D eq3}\end{multline}

With~\eqref{D eq1},~\eqref{D eq2} and~\eqref{D eq3} at hand, it can be explicitly proved that
\begin{multline} \Delta_\lambda\psi_1 \Delta_\lambda\ldots \Delta_\lambda \psi_{n-1} \Delta_\lambda\ =\\
=\ \sum_{i=0}^{n-1}\ \frac{\lambda^i (1-\lambda)}{1-\lambda^n}\ \ \psi_1\ldots\psi_i\,\Delta_{\lambda^n}\,\psi_{i+1}\ldots\psi_{n-1}\ . \label{D eq4}\end{multline}
Now, we can conclude. In fact, the right-hand side of~\eqref{D eq4} is a convex mixture of the entanglement--breaking channels~\eqref{D eq2} and~\eqref{D eq3}. Since the set $\mathbf{EBt}_d$ is convex, we deduce that
\begin{equation*} \Delta_\lambda\psi_1 \Delta_\lambda\ldots \Delta_\lambda \psi_{n-1} \Delta_\lambda\ \in\ \mathbf{EBt}_d \ .\end{equation*}

\end{proof}
\end{ex}

In order to clarify the role of the various assumptions that make the above calculation possible, it is useful to give a slight refinement of it. \\

\begin{ex}[Generalized Depolarizing Channels] \label{gen D ex} $\\$
Define a generalized depolarizing channel as
\begin{equation} \tilde{\Delta}_\lambda=\lambda I +(1-\lambda)\rho_0\, \text{Tr}\ , \label{tilde D} \end{equation}
where $\rho_0$ is a generic density matrix. Then it can be easily seen that the equalities 
\begin{equation} n(\tilde{\Delta}_\lambda)\,=\ \mathcal{N}_U(\tilde{\Delta}_\lambda)\ =\ \mathcal{N}(\tilde{\Delta}_\lambda) \label{N tilde D} \end{equation}
are still true. Indeed, the crucial equations~\eqref{D comp},~\eqref{D eq1},~\eqref{D eq2},~\eqref{D eq3} and so~\eqref{D eq4} hold even now up to the simple substitutions $\Delta_\lambda\rightarrow\tilde{\Delta}_\lambda$ and $\frac{\mathds{1}}{d}\rightarrow \rho_0$. What we loose in this case is the equivalent of the explicit equation~\eqref{n D}. We can only write an implicit expression
\begin{gather} n(\tilde{\Delta}_\lambda)\,=\ \mathcal{N}_U(\tilde{\Delta}_\lambda)\ =\ \mathcal{N}(\tilde{\Delta}_\lambda)\ =\ \Big\lceil\ \frac{\log\mu}{\log\lambda}\ \Big\rceil\ , \\
\mu\ \equiv\ \max\,\{\, 0\leq a\leq 1 :\ \tilde{\Delta}_a\in\mathbf{EBt}\, \}\, . \label{n tilde D} \end{gather}

\end{ex}

\section{Unitary vs. Non--Unitary Filtering} \label{sec U vs nU filt}

The direct $n$--index can be computed with a relatively easy and efficient algorithm. Given a channel $\phi$, we construct the Choi states $R_{\phi^n}$ and test their separability. The first $R_{\phi^n}$ which turns out to be separable corresponds exactly to $n=n(\phi)$. Even if deciding whether a given bipartite state is separable or not is very difficult (the separability problem is known to be \emph{NP--hard}~\cite{ENTAN}), one could easily get lower bounds by means of some necessary separability criteria and upper bounds by means of the sufficient criteria.

However, the situation is radically different for the filtered indices $\mathcal{N}_U$, $\mathcal{N}$. In that case there seems to be no a priori efficient algorithm allowing their calculation. Indeed, it must be remarked that the defining equations~\eqref{NU} and~\eqref{N} all involve a nontrivial optimization over the whole set of completely positive or unitary channels. Because of the potentially \emph{infinite number} of possible filtering strategies one has to check, the task of calculating the actual value of any filtered index might be impossible.

Interestingly enough, we examined the explicit class of depolarizing channels in arbitrary dimension, for which all the entanglement--breaking indices can be analytically calculated (see Example~\ref{N D ex}). The result of this calculation was clear: \emph{for a depolarizing channel $n=\mathcal{N}_U=\mathcal{N}$}. In this context, as already observed, the equality $n=\mathcal{N}_U$ has to be seen as a mere consequence of the incidental property~\eqref{D commute U}. However, one could think that the other equality $\mathcal{N}_U=\mathcal{N}$ is more fundamental. What should be the intuitive meaning of this equality?

The filtering maps appearing in Definition~\ref{EBi} play the role of preserving as much as possible the entanglement between Alice and a Bob. From the point of view of the Stinespring representation, every non-unitary filter acting on $A$ can be simulated by a unitary operation on a larger system $AE$ ($E$ being an external environment). This viewpoint has been already exploited in stating Theorem~\ref{St dil}. Anyway, because of this global unitary evolution, some of the entanglement initially present between $A$ and $B$ is wasted to create uncontrolled, apparently useless quantum correlations with $E$. This invariably weakens the link between Alice and Bob. Anyway, all that can be avoided if Alice chooses to use only unitary filters. Thanks to this discussion, one could think that \emph{the optimal filtering strategy might involve, after all, only unitary filters}.

Perhaps surprisingly, this in general false. In other words, \emph{it can happen that the best unitary filtering strategy is much less effective than a non--unitary one}. We devote the rest of this section to the construction of an explicit example of this behavior, for all dimensions $d\geq 3$. \\

\begin{ex} \label{cex UFO} $\\$
In~\cite{Werner}, Werner introduces the $(U\otimes U)$--invariant states on a bipartite $(d\times d)$--dimensional system:
\begin{equation} \chi_\varphi\equiv\,\frac{(d\varphi-1) S + (d-\varphi)\mathds{1}}{d\,(d^2-1)}\ ,\ \ -1\leq\varphi\equiv\text{\emph{Tr}}[\chi_\varphi S]\,\leq 1\ . \label{chiW states phi}\end{equation}
Here, the symbol $S$ denotes the \emph{swap operator}, defined on a bipartite system by the equation
\begin{equation*} S\ \Ket{\alpha}\otimes\Ket{\beta}\ =\ \Ket{\beta}\otimes\Ket{\alpha}\ . \end{equation*}
For the sake of simplicity, it is more convenient to make the substitution
\begin{equation*} \eta\,\equiv\, \frac{1-d\varphi}{d^2-1}\ , \end{equation*}
by means of which one has
\begin{equation} \chi_\eta\,\equiv\,-\,\eta\,\frac{S}{d}\, +\, (1+\eta)\,\frac{\mathds{1}}{d^2}\ ,\quad -\frac{1}{d+1}\leq\eta\leq\frac{1}{d-1}\ . \label{chiW states}\end{equation}

In what follows we shall adopt $\eta$ as our parameter. Remarkably, Werner proved that the precise range of $\eta$ (or $\varphi$) can be determined, for which $\chi_\eta$ is separable:
\begin{equation} \chi_\eta\ \text{is separable}\quad \Longleftrightarrow\quad -\frac{1}{d+1}\leq\eta\leq\frac{1}{d^2-1}\ . \label{chiW sep}\end{equation}

We highlight that~\eqref{chiW sep} is a great conceptual achievement, because of the intrinsic difficulties one encounters when dealing with the separability problem in generic dimension. We can move the whole power of~\eqref{chiW sep} into the world of quantum channels, thanks to the Choi--Jamiolkowski isomorphism $\phi\leftrightarrow R_\phi$. The Choi dual of~\eqref{chiW states} is
\begin{equation} V_\eta\,\equiv-\,\eta\, T\, +\, (1+\eta)\, \frac{\mathds{1}}{d}\, \text{\emph{Tr}}\ ,\quad -\frac{1}{d+1}\leq\eta\leq\frac{1}{d-1}\, . \label{V Ch}\end{equation}
This equation defines a one-parameter set of $\mathbf{CPt}_d$ quantum channels, just like~\eqref{D Ch}. We adopt the standard notation of~\eqref{dep id act}, and indicate with $T$ the matrix transposition. Moreover,~\eqref{chiW sep} becomes
\begin{equation} V_\eta\in\mathbf{EBt}_d\quad \Longleftrightarrow\quad -\frac{1}{d+1}\leq\eta\leq\frac{1}{d^2-1}\ . \label{V EB}\end{equation}

It is worth noting that these \emph{Werner channels} $V_\eta$ obey simple rules of composition, which complete~\eqref{D comp} :
\begin{equation} V_{\eta_1}V_{\eta_2} = \Delta_{\eta_1\eta_2}\ ,\quad V_\eta \Delta_\lambda = \Delta_\lambda V_\eta = V_{\lambda\eta}\ . \label{V comp}\end{equation}
Moreover, an equality analogous to~\eqref{D commute U} holds:
\begin{equation} V_\eta\,\mathcal{U}\,\equiv\,\mathcal{U}^* V_\eta\ ,\quad \forall\ \mathcal{U}\in\mathbf{U}\ . \label{V almost comm U}\end{equation}
If $\,\mathcal{U}(X)=UXU^\dag$, here we indicate with $\,\mathcal{U}^*$ the channel $\,\mathcal{U}^*(X)=U^* X U^T$ (which is nothing but the conjugation by $U^*$).

Although it is not immediately obvious, the channels~\eqref{D Ch} and~\eqref{V Ch} are unitary equivalent for the qubit case $d=2$. More precisely, the fortuitous equality
\begin{equation} d=2\quad \Rightarrow\quad \mathds{1}\text{\emph{Tr}} - I \,=\, \mathcal{Y}T  \label{dep = T q} \end{equation}
(where $\mathcal{Y}$ indicates the conjugation by the second Pauli matrix, and $T$ the matrix transposition, as usual) allows us to prove that
\begin{equation} d=2\quad \Rightarrow\quad V_\eta\,\equiv\,\mathcal{Y}\,\Delta_\eta\ . \label{V = D q} \end{equation}
Because of~\eqref{V = D q}, the qubit case does not deserve any further attention; we analyzed it in Examples~\ref{n D ex} and~\ref{N D ex}. On the contrary, for $d\geq 3$ these two sets of channels are truly different. Now, observe that
\begin{multline} d\geq 3\quad \Rightarrow\\
\Rightarrow\quad \forall\ -\frac{1}{d+1}\leq\eta\leq\frac{1}{d-1}\, ,\quad \eta^2\leq\frac{1}{d+1}\ . \label{V d>2}\end{multline}
Thanks to~\eqref{V comp} and to~\eqref{D EB}, this is the same as to say that
\begin{equation} d\geq 3\quad \Rightarrow\quad V_\eta^2\in\mathbf{EBt}_d\ . \label{V comp EB} \end{equation}

But not only: provided that $d\geq 3$,~\eqref{V almost comm U} (together with~\eqref{V comp EB}) implies that \emph{there is no unitary filter we can use in order to prevent the complete destruction of the entanglement after two iterations}. Indeed, if $\,\mathcal{U}$ is an unitary evolution,
\begin{equation} V_\eta\,\mathcal{U}\,V_\eta\ =\ \mathcal{U}^*\,V_\eta^2\ \in\ \mathbf{EBt}_d\ . \label{U filt V useless} \end{equation}
Observe that we used also~\eqref{EB propagates} in the last passage. In other words, we proved that
\begin{equation} d\geq 3\quad \Rightarrow\quad n\,(V_\eta)\,=\, \mathcal{N}_U(V_\eta)\,=\, 2\ . \label{NU V}\end{equation}

Therefore, whatever unitary filtering strategy is in the present case demonstrably useless. Let us try another kind of quantum channel as a filter. In the following we shall deal only with the extreme case $\eta=\frac{1}{d-1}$. Indeed, in that case the calculations are much simpler. Consider the Hilbert space $\mathds{C}^d$ (with $d\geq 3$) spanned by the $d$ vectors $\{\Ket{0},\,\Ket{1},\ldots,\Ket{d-1} \}$. Moreover, define the quantum channel $\psi$ whose action is
\begin{multline} \psi(\rho)\ =\ \left(\, \Ket{0}\!\!\Bra{1} + \Ket{1}\!\!\Bra{0}\, \right)\ \rho\ \left(\, \Ket{0}\!\!\Bra{1} + \Ket{1}\!\!\Bra{0}\, \right)\ +\\
+\ \sum_{i=2}^{d-1} \Ket{0}\!\!\Bra{i}\rho\Ket{i}\!\!\Bra{0}\ . \label{cex psi} \end{multline}
A more compact form of~\eqref{cex psi} can be written if one decomposes $\rho$ as a block matrix
\begin{equation*} \rho\ =\ \begin{pmatrix} A & B \\ B^\dag & C \end{pmatrix}\ \ , \end{equation*}
where $A$ and $C$ have sizes $2\times 2$ and $(d-2)\times(d-2)$, respectively, while $B$ is a $2\times(d-2)$ rectangular matrix. In that case, denoting by $X$ the first Pauli matrix, one has
\begin{equation} \psi(\rho)\ =\ \psi\,\begin{pmatrix} A & B \\ B^\dag & C \end{pmatrix}\ =\ \begin{pmatrix} XAX+\Ket{0}\!\!\Bra{0} \text{\emph{Tr}}\, C & 0 \\ 0 & 0 \end{pmatrix}\ . \label{cex psi block} \end{equation}
Observe that for every $k\geq 1$ one has
\begin{equation} \psi^{2k-1}\equiv\,\psi\ ,\quad \psi^{2k}\equiv\,\psi^2\ . \label{cex V eq1}\end{equation}
Moreover, we have simply
\begin{equation} T\psi T\ =\ \psi\ . \label{cex V eq2} \end{equation}

Now, we claim that for every $n\in\mathds{N}$ and $d\geq 3$, one has
\begin{equation} \underbrace{V_{\frac{1}{d-1}}\ \psi\ V_{\frac{1}{d-1}}\ \ldots\ V_{\frac{1}{d-1}}\ \psi\ V_{\frac{1}{d-1}}}_{\text{$V_\frac{1}{d-1}$ repeated $2n+1$ times}}\ \ \notin\ \ \mathbf{EBt}_d \label{cex t}\end{equation}
As a consequence,
\begin{equation} \mathcal{N}\left(V_{\frac{1}{d-1}}\right)\ =\ \infty\ \ . \label{N V}\end{equation}
Observe that equations~\eqref{N V} and~\eqref{NU V} \emph{explicitly prove} (for every $d\geq 3$) that the non-unitary filtering strategies can be much better than the unitary ones.

\begin{proof}[Proof of~\eqref{cex t}] $\\$
In order to prove~\eqref{cex t}, we will write the Choi matrix $R_{T\xi}$ corresponding to $T\xi$; here we have defined for short
\begin{equation*} \xi\ \equiv\ \underbrace{V_{\frac{1}{d-1}}\ \psi\ V_{\frac{1}{d-1}}\ \ldots\ V_{\frac{1}{d-1}}\ \psi\ V_{\frac{1}{d-1}}}_{\text{$V_\frac{1}{d-1}$ repeated $2n+1$ times}}\ \ . \end{equation*}
Next, we will verify that $R_{T\xi}\ngeq 0$; the PPT criterion will imply that $\xi\notin\mathbf{EBt}$, i.e. the thesis.

Firstly, write for the $V_\eta$ channels the analogous of the composition formula~\eqref{D eq1}, with the same shorthand notation as in~\eqref{dep id act} :
\begin{multline} V_\eta\psi_1 V_\eta\ldots V_\eta \psi_{k-1} V_\eta\ =\ (-\eta)^k\ T\psi_1T\ldots T\psi_{k-1}T\ +\\
+\ (1+\eta)\ \sum_{i=1}^{k-1}\, (-\,\eta)^i\, \left( T\psi_1\ldots T\psi_i \right)\left( \frac{\mathds{1}}{d} \right)\ \text{Tr}\ +\\
+\ (1+\eta)\ \frac{\mathds{1}}{d}\ \text{Tr} \ . \label{cex V eq3}\end{multline}
In our case we have $\psi_1=\ldots=\psi_{k-1}=\psi$, $k=2n+1$ and $\eta=\frac{1}{d-1}$. Because of equations~\eqref{cex V eq1} and~\eqref{cex V eq2},~\eqref{cex V eq3} becomes
\begin{multline} T\,\xi\ =\ -\,\frac{1}{(d-1)^{2n+1}}\ \psi^2\ -\\
-\, \frac{(d-1)^{2n}-1}{d(d-2)(d-1)^{2n}}\, \left( \psi(\mathds{1})\,-\,\frac{1}{d-1}\ \psi^2(\mathds{1})\right)\ \text{Tr}\ +\\
+\ \frac{\mathds{1}}{d-1}\ \text{Tr}\ . \label{cex V eq4}\end{multline}

The Choi matrix $R_{T\xi}$ is a complicated object. However, we are interested only in proving that \emph{it is not positive definite}. To this purpose, we can examine its restriction to the subspace spanned by $\{\, \Ket{00},\,\Ket{11}\,\}$. Thanks to the properties of the Choi--Jamiolkowski isomorphism, we have
\begin{equation} \Braket{ij|R_{T\xi}|kl}\ =\ \frac{1}{d}\ \Braket{i|\, (T\xi)\,(\Ket{j}\!\!\Bra{l})\,|k}\ \ . \end{equation}
By applying repeatedly this identity and~\eqref{cex psi block}, one can see that
\begin{equation*} \Braket{00|R_{T\xi}|00}\, =\, 0 \, ,\quad \Braket{00|R_{T\xi}|11}\, =\, -\frac{1}{d\,(d-1)^{2n+1}}\ . \end{equation*}
Therefore, there exists $a\in\mathds{R}$ such that
\begin{equation*} R_{T\xi}\big|_{\text{Span}\, \{ \Ket{00},\,\Ket{11} \}}\, =\ \frac{1}{d\,(d-1)^{2n+1}}\ \begin{pmatrix} 0 & -1 \\ -1 & a \end{pmatrix}\ \ . \end{equation*}
Since
\begin{equation*} \det\ \begin{pmatrix} 0 & -1 \\ -1 & a \end{pmatrix}\ =\ -1 \ , \end{equation*}
the restriction $R_{T\xi}\big|_{\text{Span}\, \{ \Ket{00},\,\Ket{11} \}}$ can not be positive definite. This necessarily forbids $R_{T\xi}\geq 0$, and so $T\xi\notin\mathbf{CPt}$. Thanks to the PPT criterion, we can conclude that $\xi\notin\mathbf{EBt}$, i.e.~\eqref{cex t}.

\end{proof}
\end{ex}

Let us make the main point one more time. Example~\ref{cex UFO} shows that the optimal filtering strategy to be used by Alice against the local noise can be, as a matter of fact, \emph{non-unitary}. From the physical point of view, we are claiming that Alice can be forced to introduce other (controlled) disturbances into her system, so as to save the entanglement with Bob. Moreover, equations~\eqref{NU V} and~\eqref{N V} show that the difference between the best unitary strategy and the best non-unitary one can be dramatic. The former causes the almost immediate destruction of the entanglement, while the latter allows its unlimited survival. 
A similar non--optimality of unitary filtering operations was also observed in the context of quantum subdivision capacities~\cite{MULLER}. All that may appear quite counterintuitive. However, a possible (albeit not rigorous) physical justification can be found by invoking an argument which has been introduced in the study of optimal recovery transformations~\cite{PATA} where, assigned a given dynamical semigroup, one is asked to identified the best output quantum data--process that guarantees that the average output  fidelity is maximal. 
The idea is as follows. Since quantum channels represent the occurrence of stochastic errors, in general they will tend to pump in entropy into the system (heating and diffusional processes) or, vice-versa, pump out entropy from the system (dissipative or cooling processes). In both cases, to fight such effects one needs to  modify the entropy content of the state, i.e. again with  dissipative or heating processes:
unitary recovery operations appear not to be well suited for this purpose, as the best they can do is to concentrate the entropy extras or deficits into a specific subsystem, without removing them.

\section{Filtered Indices for Qubit Channels} \label{sect filt q}

An amazing fact about the Example~\ref{cex UFO} in that it works only for $d\geq 3$. This restriction comes from~\eqref{V comp EB}, and instills in us a glimmer of hope that things could be different, after all, for $d=2$. For this reason, this Section is devoted to the investigation of the qubit case. In fact, the Bloch representation~\eqref{Bloch repr} can considerably simplify the theory for two-dimensional systems. An explicit example of this simplification is presented in Subsection~\ref{subsect NU u q}. We continue our analysis by proposing a general conjecture in Subsection~\ref{subsect conj}. Finally, some partial proofs of this conjecture are showed in Subsection~\ref{subsect proof infty} and~\ref{subsect hard proof}.

\subsection{Unitary Filtered Index for Unital Qubit Channels} \label{subsect NU u q}

We begin by translating in our language and notation a result originally proved in~\cite{V} (although in a slightly weaker form). All that will solve the simplest problem of the calculation of $\mathcal{N}_U$ for an unital qubit channel.

For the sake of clearness, we firstly recall some facts about the canonical diagonal form for a qubit quantum channel $\phi=(M,c)$. For details, we refer the reader to~\cite{KingRuskai} and~\cite{RuskaiCPqubit}. Let $M=P D Q$ be a singular value decomposition of $M$, with $P,Q$ orthogonal matrices. Denoting by~$\{ s_i\}$ the singular values of $M$, we have
\begin{equation*} D\, =\, \begin{pmatrix} s_1 & 0 & 0 \\ 0 & s_2 & 0 \\ 0 & 0 & s_3 \end{pmatrix}\, .  \end{equation*}
In order to give a physical interpretation to this algebraic decomposition, it is not sufficient that $P,Q$ are orthogonal, but it is necessary that \mbox{$P,Q\in\text{SO}(3)$} (i.e. they must be \emph{special} orthogonal). Suppose that this does not happen, and examine the other cases. If \mbox{$\det P= \det Q=-1$} (and so $\det M\geq0$) we can simply write $M=(-P)D(-Q)$, in such a way that \mbox{$\det(-P)=\det(-Q)=+1$} and therefore \mbox{$-P,-Q\in \text{SO}(3)$}. On the other hand, if \mbox{$\det P=-1=-\det Q$} or the converse (and so \mbox{$\det M\leq 0$}), we must modify $D$ and write for example $M=\tilde{P}\tilde{D}Q$, with
\begin{equation*} \tilde{D} \equiv \begin{pmatrix} s_1 & 0 & 0 \\ 0 & s_2 & 0 \\ 0 & 0 & -s_3 \end{pmatrix}\, ,\qquad \tilde{P} \equiv P \begin{pmatrix} 1 & 0 & 0 \\ 0 & 1 & 0 \\ 0 & 0 & -1 \end{pmatrix} \in \text{SO}(3)\, . \end{equation*} 
This discussion should convince the reader that the best \emph{special singular value decomposition} we can achieve is of the form $M=O_1 L O_2$, with $O_1,O_2\in\text{SO}(3)$ and
\begin{multline} L = \begin{pmatrix} l_1 & 0 & 0 \\ 0 & l_2 & 0 \\ 0 & 0 & l_3 \end{pmatrix} \equiv \begin{pmatrix} s_1 & 0 & 0 \\ 0 & s_2 & 0 \\ 0 & 0 & \text{sgn} \det (M)\ s_3 \end{pmatrix} . \label{k} \end{multline}
Here the symbol $\text{sgn}$ denotes the \emph{sign function}, defined by
\begin{equation*} \text{sgn}\ x\, \equiv\, \left\{ \begin{array}{cr} +1 & \text{ if $x>0$}\\ 0 & \text{ if $x=0$}\\ -1 & \text{ if $x<0$} \end{array} \right.\ \ . \end{equation*}
Usually we shall suppose $|s_3|\leq s_1,s_2$, so that $l_1,l_2\geq 0$ and only $l_3$, which has the lowest modulus, can be negative. Once the decomposition $M=O_1 L O_2$ is obtained, we can define $t\equiv O_1^T c$ and write 
\begin{equation} \phi\ =\ (M,c)\ =\ O_1\ (L, t)\ O_2\ =\ \mathcal{U}\ \Lambda\ \mathcal{V}\ . \label{q canonical form} \end{equation}
Here $\mathcal{U},\mathcal{V}$ are the unitary channels corresponding to $O_1,O_2\in\text{SO}(3)$, and $\Lambda\equiv(L,t)$ is the \emph{canonical diagonal form} of $\phi$. \\

\begin{thm} \label{NU=mU unital q} $\\$
Let $(M,0)\in\mathbf{CPt}_2$ be an unital qubit channel. Denote by $M=O_1 L O_2$ the special singular value decomposition of $M$. Then
\begin{multline} \mathcal{N}_U(M,0)\, =\, n(L,0)\, =\\
=\, \min{\{\, n\geq 1\, :\ \sum_{i=1}^3 |l_i|^n\leq 1\, \}}\, =\\
=\, \min{\{\, n\geq 1\, :\ \|M\|_n\leq 1\, \}}\, . \label{mU=NU unital q} \end{multline}
\end{thm}

\begin{proof}
The explicit expressions for $n(L,0)$ are direct consequences of~\eqref{nfu2}, and of the elementary observation
\begin{equation*} \sum_{i=1}^3 |l_i|^n\, =\, \|L\|_n\, =\, \|M\|_n\, . \end{equation*}
On the other hand, elementary properties~\eqref{NUC} and \eqref{elem ineq} ensure that
\begin{equation*} n(L,0) \leq\, \mathcal{N}_U (L,0) =\, \mathcal{N}_U (M,0)\, . \end{equation*}
Consequently, the only nontrivial claim is that \mbox{$n(L,0)\geq\mathcal{N}_U(L,0)$}, so that the inequality in the previous equation is actually an equality. Thanks to~\eqref{nfu2}, we have only to prove the $p=1$ case of the following statement:
\begin{multline} \forall\, n\geq 1\, ,\ \ \forall\, \ O_1, \ldots, O_n\in \text{SO}(3)\, ,\\
\|LO_1L\ldots L O_n L \|_p\,\leq\, \| L^{n+1} \|_p \ . \label{NmUq eq1}\end{multline}
Indeed,~\eqref{NmUq eq1} would imply that the channel $(LO_1L\ldots L O_n L,\ 0)$ must necessarily be entanglement--breaking if so is $L^{n+1}$ .

In what follows, we will use extensively the well-known H\"{o}lder inequality
\begin{multline} \frac{1}{r}+\frac{1}{s}=1\ ,\ \ 1\leq p\leq \infty\ \ \Rightarrow\\
\Rightarrow\ \ \|AB\|_p\, \leq\, \| |A|^r \|_p^{1/r} \, \| |B|^s \|_p^{1/s}\, =\\
=\, \|A\|_{rp} \, \|B\|_{sp}\, . \label{Hoelder}\end{multline}

Since $L$ is diagonal, observe that for every $1\leq p\leq\infty$ and for every integer $n\geq 1$ we can write
\begin{equation} \|L\|_{np} = \|L^n\|_p^{1/n} \label{L diag H}\end{equation}
Then, the best way to prove~\eqref{NmUq eq1} is by induction.
\begin{itemize}
\item For $n=1$, thanks to the $r=s=2$ case of~\eqref{Hoelder} we have
\begin{gather*} \|L (O L) \|_p\ \leq\ \|L\|_{2p}\  \|O L\|_{2p}\ =\\
=\ \|L\|_{2p}\ \|L\|_{2p}\ =\ \|L\|_{2p}^2\ =\ \|L^2\|_p\ , \end{gather*}
where we used the unitary invariance of the Schatten norms, together with~\eqref{L diag H}.

\item Now, suppose that we have proved the inequality for every $p$ and for $n-1$; we can apply H\"{o}lder again with $r=n+1,\ s=\frac{n+1}{n}$, obtaining
\begin{gather*}
\|L O_1 L \ldots L O_n L \|_p\ \leq\\
\leq\ \|L\|_{(n+1)p} \ \|O_1L \ldots O_n L\|_{\frac{n+1}{n}p}\ =\\
=\ \|L\|_{(n+1) p}\ \|L O_2 L \ldots L O_n L\|_{\frac{n+1}{n}p}\ \leq\\
\leq\ \|L\|_{(n+1)p}\ \|L^n \|_{\frac{n+1}{n}p}\ =\\
=\ \|L^{n+1}\|_p^\frac{1}{n+1}\ \|L^{n+1} \|_p^{\frac{n}{n+1}}\ =\ \|L^{n+1}\|_p\ .
\end{gather*}
We used, in order,~\eqref{Hoelder}, the unitary invariance of the Schatten norms, the inductive hypothesis, and~\eqref{L diag H}.
\end{itemize}
\end{proof}

It is worth noting that Proposition~\ref{NU=mU unital q} gives us a simple procedure to calculate the unitary filtered index $\mathcal{N}_U$ at least in the simplest case of unital qubit channels. In spite of the strict restriction it is subjected to,~\eqref{mU=NU unital q} is quite encouraging. In fact, it shows how the theory of the filtered indices can be simpler in the qubit case than in general, because of the low dimensionality of the system under examination.

\subsection{Non--Unitary Filtered Indices for Qubit: a Conjecture} \label{subsect conj}

We have already observed that Example~\ref{cex UFO}, showing the non--optimality of purely unitary fitering strategies, works only in dimension $d\geq 3$. Consequently, one could ask himself, whether or not the unitary filters are optimal for qubit channels. This is the main content of the following conjecture.

\begin{cj} \label{UFO q} $\\$
Let $\phi\in\mathbf{CPt}_2$ be a qubit channel. Then
\begin{equation} \mathcal{N}(\phi)\ =\ \mathcal{N}_U(\phi)\ =\ \max{\{\, n\,(\mathcal{U}\phi)\, : \ \mathcal{U}\in\mathbf{U}_2\, \}} \ . \label{UFO eq} \end{equation}
\end{cj}

Actually, Conjecture~\ref{UFO q} claims something more, that is, it gives an explicit algorithm to calculate $\mathcal{N}_U$ (and so $\mathcal{N}$) in terms of a single optimization over the set of unitary operations. This corresponds to say that the best unitary filtering strategy involves the iteration of a \emph{single} unitary filter.

\subsection{Divergent Filtered Indices for Qubit} \label{subsect proof infty}

An equality such as $\mathcal{N}_U=\mathcal{N}$, which is the heart of Conjecture~\ref{UFO q}, is explicitly violated in dimension $d\geq 3$. Example~\ref{cex UFO} shows that this violation can be dramatic, with $\mathcal{N}_U=2$ and $\mathcal{N}=\infty$. However, we will be able to show that \emph{such an extreme possibility can be ruled out} in the qubit case. It must be remarked that this constitutes only a partial (though encouraging) proof of the conjecture under examination. 

Actually, here we will present two different partial proofs. On one hand, the first result states that $\mathcal{N}(\phi)=\infty$ is possible only if also $\mathcal{N}_U(\phi)=\infty$, for all qubit channels $\phi$; on the other hand, the second result claims that $\mathcal{N}_U(\phi)=2$ implies $\mathcal{N}(\phi)=2$, for all unital qubit channels $\phi$. \\

\begin{thm}[Proof of Conjecture~\ref{UFO q} for Qubit Channels with~$\mathcal{N}_U=\infty$] \label{proof infty}
Let $\phi\in\mathbf{CPt}_2$ be a qubit channel. Then the following are equivalent.

\begin{enumerate}

\item $\max{\{\, n\,(\mathcal{U}\phi)\, : \ \mathcal{U}\in\mathbf{U}_2\, \}}\, =\, \infty$ .

\item $\mathcal{N}_U(\phi)=\infty$ .

\item $\mathcal{N}(\phi)=\infty$ .

\item The image of the Bloch sphere under the action of $\phi$ contains a pure state, and $\phi$ is not entanglement--breaking.

\end{enumerate}
\end{thm}

\begin{proof} $\\$
\begin{description}

\item[$1 \Rightarrow 2$ :] From the very Definition~\ref{EBi} it follows that
\begin{gather*}
\mathcal{N}_U(\phi)\ =\ \min{\{\,n\geq 1:\ \forall \ \mathcal{U}_1,\ldots,\mathcal{U}_{n-1} \in \mathbf{U}_2,}\\
\phi\,\mathcal{U}_1\phi\ldots\phi\,\mathcal{U}_{n-1}\phi \in \mathbf{EBt}_2\, \}\ \geq\\
\geq\ \min{\{\,n\geq 1:\ \forall \ \mathcal{U} \in \mathbf{U}_2, \underbrace{\phi\,\mathcal{U}\phi\ldots\phi\,\mathcal{U}\phi}_{\text{$\phi$ repeated $n$ times}} \in \mathbf{EBt}_2\, \}}\ =\\
=\ \min{\{\,n\geq 1:\ \forall \ \mathcal{U} \in \mathbf{U}_2,\ (\mathcal{U}\phi)^n \in \mathbf{EBt}_2\, \}}\ =\\
=\ \max{\{\, n\,(\mathcal{U}\phi)\, : \ \mathcal{U}\in\mathbf{U}_2\, \}} \ . 
\end{gather*}
Therefore, $\max{\{\, n\,(\mathcal{U}\phi)\, : \ \mathcal{U}\in\mathbf{U}_2\, \}}\, =\, \infty$ directly implies $\mathcal{N}_U(\phi)=\infty$.

\item[$2 \Rightarrow 3$ :] This implication follows from~\eqref{elem ineq}.

\item[$3 \Rightarrow 4$ :] Obviously, if $\mathcal{N}(\phi)=\infty$ then $\phi$ can not be entanglement--breaking. It remains to show that the image of the Bloch sphere through $\phi$ always contains a pure state. In what follows, we will denote by $(M,c)$ the Bloch sphere representation~\eqref{Bloch repr} of $\phi$.

If $\|M\|_\infty=1$ we can immediately conclude, because of the following reasoning. Take an unit vector $\vec{r}$ such that $|M\vec{r}|=\|M\|_\infty=1$; since $\phi=(M,c)$ is a positive map, it must be $|M\vec{r}\pm c|\leq 1$, and so
\begin{multline} 1\ \geq\  \frac{1}{2}\,\left(\, |M\vec{r} + \vec{c}|^2 + |M\vec{r} - \vec{c}|^2\, \right)\ =\\
=\ \|M\|_\infty^2 +\, |\vec{c}|^2\ . \label{infty 1} \end{multline}
If $\|M\|_\infty=1$,~\eqref{infty 1} implies that $\vec{c}=0$, and so that 
\begin{equation*} \phi\,\left(\,\frac{\mathds{1}+\vec{r}\cdot\vec{\sigma}}{2}\,\right)\ =\ \frac{\mathds{1}+(M\vec{r})\cdot\vec{\sigma}}{2} \end{equation*}
is a pure state.

Now, let us suppose $\|M\|_\infty<1$. By hypothesis,
\begin{multline}
\forall\, n\geq 1\, ,\ \ \exists\ \psi_1^{(n)} ,\ldots, \psi_{n-1}^{(n)}\, \in \mathbf{CPt}_2\ :\\
\phi\psi_1^{(n)}\phi\ldots\phi\psi_{n-1}^{(n)}\phi \,\notin\, \mathbf{EBt}_2 . \label{infty 2}
\end{multline}
Consider the sequence of qubit channels \mbox{$\left( \psi_1^{(n)}\phi\ldots\psi_{n-1}^{(n)}\phi \right)_{n\geq 1}$}; since its elements belong to the compact set $\mathbf{CPt}_2$, it must admit a limit point $\chi\in\mathbf{CPt}_2$. Moreover, $\chi$ must be of the form $\chi=(0,s)$ (with $|s|\leq 1$), because of the assumption $\|M\|_\infty<1$. In fact, using the notation $\psi_i^{(n)}=(N_i^{(n)},\ldots)$, we have
\begin{equation*} \psi_1^{(n)}\phi\ldots\psi_{n-1}^{(n)}\phi\ =\ (\,N_1^{(n)} M \ldots N_{n-1}^{(n)} M\,,\ \ldots\,)\ , \end{equation*}
and
\begin{equation*} \|\,N_1^{(n)} M \ldots N_{n-1}^{(n)} M\, \|_\infty\ \leq\ \|M\|_\infty^{n-1} \xrightarrow[n\rightarrow\infty]{} 0 . \end{equation*}
Observe that we have used the bound \mbox{$\|N_i^{(n)}\|_\infty\leq 1$}, which descends directly from the analogous of~\eqref{infty 1}.

By the very definition of $\chi$, it follows that the sequence of non--entanglement--breaking channels \mbox{$\left( \phi\psi_1^{(n)}\phi\ldots\psi_{n-1}^{(n)}\phi \right)_{n\geq 1}$} admits the limit point $\phi\chi=(0,\,Ms+c)\in\mathbf{CPt}_2$. Its Choi state
\begin{gather*} R_{\phi\chi}\ =\ (\phi\chi \otimes I)(\Ket{\varepsilon}\!\!\Bra{\varepsilon})\ =\\
=\ \frac{\mathds{1}+(M\vec{s}+\vec{c})\cdot\vec{\sigma}}{2}\,\otimes\,\frac{\mathds{1}}{2}\ =\ \phi\,\left( \frac{\mathds{1}+\vec{s}\cdot\vec{\sigma}}{2} \right)\,\otimes\,\frac{\mathds{1}}{2} \end{gather*}
must belong to the \emph{boundary} of the set of separable states, otherwise $\phi\psi_1^{(n)}\phi\ldots\psi_{n-1}^{(n)}\phi$ would be entanglement--breaking for sufficiently large $n$. This means that $R_{\phi\chi}=R_{\phi\chi}^{T_B}$ can not be \emph{strictly} positive definite (see equation~(15.65) of~\cite{GeometryQuantum}). Instead, it must have at least a zero eigenvalue, that is, $\phi\,\left(\, (\mathds{1}+\vec{s}\cdot\vec{\sigma})/2\, \right)$ must be a pure state. In conclusion, we have found a pure state in the image of the Bloch sphere through $\phi$.

\item[$4 \Rightarrow 1$ :] We can suppose without loss of generality that $\phi\,(\Ket{0}\!\!\Bra{0})$ is a pure state. Consider an unitary conjugation $\,\mathcal{U}\in\mathbf{U}_2$ such that $\mathcal{U}\phi\,(\Ket{0}\!\!\Bra{0})=\Ket{0}\!\!\Bra{0}$; we will prove that $n\,(\mathcal{U}\phi)=\infty$.

Let $\{ M_k \}_k$ be a Kraus representation of the qubit channel $\,\mathcal{U}\phi$. The equality $\mathcal{U}\phi\,(\Ket{0}\!\!\Bra{0})=\Ket{0}\!\!\Bra{0}$ is possible if and only if $M_k\Ket{0}\propto\Ket{0}$ for all $k$. Then there exists a pure state $\Ket{\alpha}$ such that $\mathcal{U}\phi\,(\Ket{0}\!\!\Bra{1})\propto\Ket{0}\!\!\Bra{\alpha}$; because of the trace--preserving property, we can chose $\Ket{\alpha}=\Ket{1}$, so that $\mathcal{U}\phi\,(\Ket{0}\!\!\Bra{1})= z \Ket{0}\!\!\Bra{1}$ for some $z\in\mathds{C}$. Taking the hermitian conjugate, we obtain also $\mathcal{U}\phi\,(\Ket{1}\!\!\Bra{0}) = z^* \Ket{1}\!\!\Bra{0}$. Observe that it must be $z\neq 0$, otherwise we could write for all $2\times 2$ matrices $X$
\begin{equation*} \mathcal{U}\phi(X)\, =\, \Ket{0}\!\!\Bra{0}\, \Braket{0|X|0}\, +\, \mathcal{U}\phi(\Ket{1}\!\!\Bra{1})\, \Braket{1|X|1}\ , \end{equation*}
and $\,\mathcal{U}\phi$ (that is, $\phi$) would be entanglement--breaking, by comparison with~\eqref{Holevo form}. This is explicitly forbidden by hypothesis, and we must conclude that $z\neq 0$. By iteration, the following equalities are valid for all $n\geq 1$.
\begin{gather*} (\mathcal{U}\phi)^n\,(\Ket{0}\!\!\Bra{0})\,=\,\Ket{0}\!\!\Bra{0} \\
(\mathcal{U}\phi)^n\,(\Ket{0}\!\!\Bra{1})\, =\, z^n \Ket{0}\!\!\Bra{1} \\
(\mathcal{U}\phi)^n\,(\Ket{1}\!\!\Bra{0})\, =\, (z^*)^n \Ket{1}\!\!\Bra{0}\, . \end{gather*}
As a consequence, representing $(I\otimes (\mathcal{U}\phi)^n)\,(\Ket{\varepsilon}\!\!\Bra{\varepsilon})$ in the lexicographically ordered basis $\{\, \Ket{00},\, \Ket{01},\, \Ket{10},\, \Ket{11}\,\}$ we obtain
\begin{equation} (I\otimes (\mathcal{U}\phi)^n)\,(\Ket{\varepsilon}\!\!\Bra{\varepsilon})\ =\ \begin{pmatrix} 1 & 0 & 0 & z \\ 0 & 0 & 0 & 0 \\ 0 & 0 & a_n & b_n \\ z^* & 0 & b_n^* & 1-a_n \end{pmatrix}\, . \label{infty 3} \end{equation}
Since $z\neq 0$, the application of the PPT criterion reveals that~\eqref{infty 3} is an entangled state. Therefore, $(\mathcal{U}\phi)^n$ is \emph{not} entanglement--breaking for all $n\geq 1$, that is, $n\,(\mathcal{U}\phi)=\infty$.

\end{description}
\end{proof}

While the proof of Theorem~\ref{proof infty} is rather complicated, its meaning is pretty much clear. Specifically, at variance with the higher dimensional case, for qubit channels it is impossible to have  $\mathcal{N}_U=2$ while \mbox{$\mathcal{N}=\infty$} (see Example~\ref{cex UFO}): instead if one filtered index reaches $\infty$, then the same happens to the other.

\subsection{Conjecture~\ref{UFO q} for Qubit: Simple Unital Case} \label{subsect hard proof}

The analysis presented in the previous section is clearly not sufficient to prove that  Conjecture~\ref{UFO q} is true (for instance, it is still possible that there exist counterexamples satisfying  $\mathcal{N}_U=2$ but $3\leq\mathcal{N}<\infty$). The restriction $\mathcal{N}_U=2$ is reasonable, because it confines our analysis to the first nontrivial case. However, we will show that at least \emph{for unital qubit channels, if $\mathcal{N}_U=2$, then also $\mathcal{N}=2$}. As a consequence, there exists no unital counterexample to Conjecture~\ref{UFO q} which satisfies the restriction $\mathcal{N}_U=2$, as Example~\ref{cex UFO} did in the case of higher dimensions.

To prove  this result we need a series of  preliminary lemmas. In particular since we are looking for a \emph{upper bound} on $\mathcal{N}$, by the very definition~\eqref{N} we will have to prove that a certain sequence of channels is entanglement--breaking. For this reason, the first task is to formalize a \emph{sufficient} separability criterion which is capable to detect the \emph{absence} of entanglement between two--dimensional systems. We use here a particular case of Proposition 3 in~\cite{J}. For the sake of clearness a simple proof is also given.

\begin{prop} \label{TN cr} $\\$
With the notation of~\eqref{Bloch repr} and~\eqref{Fano}, one has
\begin{equation} \|M\|_1 + |c|\, \leq\, 1 \quad\Rightarrow\quad (M,c) \in \mathbf{EBt}_2\ . \label{TN cr} \end{equation}
\end{prop}

\begin{proof}
By applying unitary evolutions to the left and to the right, we can suppose $(M,c)$ reduced in canonical form $(L,t)$ (see equations~\eqref{q canonical form} and~\eqref{k}).  Then, we have only to prove that $R_{(L,t)}$ is separable if $\|L\|_1+|t|\leq 1$. Write the partial transpose of $R_{(L,t)}$ as
\begin{equation} R_{(L,t)}^{T_B}\, =\, \frac{1}{4}\, \left(\, \mathds{1} + (\vec{t}\cdot\vec{\sigma})\otimes \mathds{1}\, +\, \sum_{i=1}^3 l_i\, \sigma_i \otimes \sigma_i\, \right)\ . \label{TN cr 1}\end{equation}
If we could demonstrate that $R_{(L,t)}^{T_B}\geq 0$, then the PPT condition for separability in $2\times 2$ systems would conclude the proof. This will be proved by showing that
\begin{equation*} \left\|\, (\vec{t}\cdot\vec{\sigma})\otimes \mathds{1}\, +\, \sum_{i=1}^3 l_i\, \sigma_i \otimes \sigma_i\, \right\|_\infty\ \leq\ 1\ . \end{equation*}
Firstly, observe that the following elementary equalities hold:
\begin{equation*} \|A\otimes B\|_\infty=\|A\|_\infty\,\|B\|_\infty\, ,\quad \|n\cdot\vec{\sigma}\|_\infty=|n|\, . \end{equation*}
Then, thanks to the triangular inequality, we have
\begin{gather*} \left\|\, (\vec{t}\cdot\vec{\sigma})\otimes \mathds{1}\, +\, \sum_{i=1}^3 l_i\, \sigma_i \otimes \sigma_i^T\, \right\|_\infty\ \leq\\
\leq\ \|\, (\vec{t}\cdot\vec{\sigma})\otimes \mathds{1}\,\|_\infty\, +\, \sum_{i=1}^3 |l_i|\, \| \sigma_i \otimes \sigma_i^T\|_\infty \ =\\
=\ \|\, \vec{t}\cdot\vec{\sigma}\,\|_\infty\, +\, \sum_{i=1}^3 |l_i|\, \| \sigma_i\|_\infty\, \| \sigma_i^T\|_\infty\ =\\
=\ |t|\, +\, \sum_{i=1}^3 |l_i|\ =\ |t|\, +\, \|L\|_1\ =\ |c|\, +\, \|M\|_1\ \leq\ 1\ . \end{gather*}
\end{proof}

Another technical result that will be useful through the rest of the section is the following. \\

\begin{lemma} \label{strano} $\\$
Given a vector $v\in\mathds{R}^n$, let us denote by $\left[ v \right]\in \mathds{R}^n$ the vector obtained by taking the absolute value of the components of $v$, i.e. \mbox{$\left[ v\right]_i\equiv |v_i|$}. We claim that for each $v\in\mathds{R}^n$ and for each $n \times n$ real matrix $A$, there exists a special orthogonal matrix $O\in SO(n)$ such that the two vectors
\begin{equation*} \left[ Ov \right]\ ,\quad \left(\ |(OA)_1|\ , \ldots,\ |(OA)_n|\ \right) \end{equation*}
are linearly dependent. We denote by the symbol $M_i$ the $i$th row of the matrix $M$, as usual.
\end{lemma}

\begin{proof}
In what follows we will use the notation $\sigma(X)$ to indicate the spectrum of the matrix $X$, with each eigenvalue repeated a number of times equal to its multiplicity. Now, note immediately that we are free to suppose $v,A\neq 0$, and so also (up to a rescaling constant)
\begin{equation} |v|=1=\|A\|_2\ . \label{strano eq1}\end{equation}
Now, we prove that there exists $O\in\text{SO}(n)$ such that
\begin{equation}  |(Ov)_i|\,\equiv\, |(OA)_i|\qquad \forall\ 1\leq i\leq n\ . \label{strano eq2}\end{equation}
Actually, we can freely extend the range of $O$ to all the orthogonal matrices, not necessarily with determinant equal to $+1$. Indeed, if we find an $O\in O(n)$ such that $\det O=-1$ and~\eqref{strano eq2} is satisfied, changing the sign of the first row of $O$ produces a special orthogonal matrix $O'\in SO(n)$ which verifies again $|(O'v)_i|\equiv |(O'A)_i|$.
Now, one can square~\eqref{strano eq2}, obtaining the requirement that
\begin{equation} \exists\ O\in O(n)\ :\quad \left(\,O\ (vv^T - AA^T)\ O^T\, \right)_{ii}\equiv\, 0\ . \label{strano eq3} \end{equation}
We can suppose without loss of generality that \mbox{$vv^T - AA^T$} (which is a symmetric matrix) is diagonal. Indeed, the spectral theorem guarantees that it can be diagonalized by means of an orthogonal transformation. Thus, \mbox{$vv^T - AA^T$} can be taken diagonal up to a change of variables in $O(n)$. Therefore, the set of matrices of the form $O\, (vv^T - AA^T)\, O^T$ is composed of all the symmetric matrices $S$ with spectrum $\sigma(vv^T-AA^T)$. We have to prove that at least one of them has all of the diagonal entries equal to zero. Now, we invoke Theorems~4.3.26 (p. 193) and~4.3.32 (p. 196) of~\cite{HJ1}. Their content is precisely that the condition we are looking for can be satisfied if and only if
\begin{equation*} \sigma(vv^T - AA^T)\ \prec\ \{ 0 \}\ , \end{equation*}
where the symbol $\prec$ represents the relation of \emph{majorization}. A vector $q\in\mathds{R}^n$ is said to majorize another vector $p\in\mathds{R}^n$ (and we write $p\prec q$) if the following relations are satisfied:
\begin{equation} \sum_{i=1}^k p_i^\uparrow\, \leq\, \sum_{i=1}^k q_i^\uparrow\quad \forall\ k=1,\ldots,n\, ,\qquad \sum_{i=1}^n p_i^\uparrow\, =\, \sum_{i=1}^n q_i^\uparrow\, , \label{eq majorization} \end{equation}
where $p_i^\uparrow$ is the vector obtained from $p$ by sorting its entries in ascending order.
This requirement is satisfied precisely because 
\begin{equation*} \text{Tr}\,[\,vv^T - AA^T]\ =\ |v|^2 - \|A\|_2^2\ =\ 1-1\ =\ 0\ , \end{equation*}
where we used~\eqref{strano eq1}. Therefore, we can conclude.
\end{proof}

Now, let us present the last mathematical lemma. In what follows, $|\cdot|$ stands for the absolute value of a real number or for the Euclidean norm of a $3$--vector (row or column).

\begin{lemma} \label{final} $\\$
Let $(M,c)\in\mathbf{Pt}_2$ be a (not necessarily completely) positive, trace--preserving qubit map. Then, for all $3\times 3$ real matrices $K$, we have
\begin{equation} |Kc|\ +\ \|KM\|_2\ \leq\ \|K \|_2\ . \label{final} \end{equation}
\end{lemma}

\begin{proof}
Fistly, let us show that
\begin{equation} |n^T c|\,+\, |n^T M|\, \leq\, |n|\quad \forall\ n\in\mathds{R}^3 \label{final 1}\end{equation}
if $(M,c)$ is positive. In order to prove~\eqref{final 1}, observe that we can restrict our analysis to the case $n^T c \geq 0$, up to the exchange $n\leftrightarrow -n$. Now, the positivity condition for $(M,c)$ implies that
\begin{multline*} |\vec{r}|\leq 1\quad \Rightarrow\quad \mathds{1}+\vec{r}\cdot\vec{\sigma}\,\geq\, 0\quad\Rightarrow\\
\Rightarrow\quad \mathds{1}+(M\vec{r}+\vec{c})\cdot\vec{\sigma}\, =\, \phi\,(\mathds{1}+\vec{r}\cdot\vec{\sigma})\,\geq\, 0\quad \Rightarrow\\
\Rightarrow\quad |M\vec{r}+\vec{c}|\,\leq\, 1\ . \end{multline*}
As a consequence, $n^T(M r+c)\,\leq\, |n|$ for all $n\in\mathds{R}^3$. Assuming $n^T c\geq 0$ and taking $r$ as the unit vector parallel to $n^T M$, we obtain exactly~\eqref{final 1}.

Now, consider a generic special orthogonal matrix \mbox{$O\in SO(3)$}, and apply~\eqref{final 1} with $n$ equal to the $i$th row of the matrix $OK$, denoted by $(OK)_i$. We have
\begin{equation*} |(OK)_i\, c|\, +\, |(OK)_i\, M|\, \leq\, |(OK)_i|\ , \end{equation*}
that is,
\begin{equation} |O_i\, Kc|\, +\, |(OKM)_i|\, \leq\, |(OK)_i|\ . \label{final 2}\end{equation}
Here $N_i$ denotes the $i$th row of a matrix $N$, as usual. Squaring and adding~\eqref{final 2} for $i=1,2,3$, one obtains
\begin{equation*} |Kc|^2\, +\, \| KM \|_2^2\, +\, 2\, \sum_{i=1}^3 |(OKc)_i|\, |(OKM)_i|\ \leq\ \|K\|_2^2 \end{equation*}

We will show that the sum in the expression above can be reduced to the product $|Kc|\, \|KM\|_2$. In fact, we use Lemma~\ref{strano} to choose an orthogonal matrix $O$ such that
\begin{equation*} \left[ OKc \right] \quad \text{and}\quad \left(\, |(OKM)_1|\, ,\ |(OKM)_2|\, ,\ |(OKM)_3|\, \right) \end{equation*}
are linearly dependent vectors. In that case, the equality sign holds in the Cauchy--Schwartz inequality for their scalar product, yielding exactly
\begin{gather*} \sum_{i=1}^3 |(OKc)_i|\, |(OKM)_i|\ \equiv\\
\equiv\ \left(\, \sum_{i=1}^3 (OKc)_i^2\, \right)^{1/2}\left(\, \sum_{i=1}^3 |(OKM)_i|^2\, \right)^{1/2}\ \equiv\\
\equiv\ |Kc|\, \|KM\|_2\ . \end{gather*} 
\end{proof}

Finally, we can state the main result of this subsection: \\

\begin{thm}[Proof of Conjecture~\ref{UFO q} for Unital Qubit Channels with $\mathcal{N}_U = 2$] \label{hard proof}
Let $\phi=(M,0)\in\mathbf{CPt}_2$ be an unital qubit channel such that $\mathcal{N}_U(\phi)=2$. Then also $\mathcal{N}(\phi)=2$. Formally,
\begin{equation} \forall\ \phi\in\mathbf{CPtu}_2\ ,\quad \mathcal{N}_U(\phi)\,=\,2\ \, \Longleftrightarrow\ \, \mathcal{N}(\phi)=\,2 \label{hard proof eq} \end{equation}
\end{thm}

\begin{proof}
Denote by $L$ the canonical form of $M$ (see~\eqref{q canonical form} and~\eqref{k}). Then, the fact that $\mathcal{N}_U(\phi)=2$ can be translated, by~\eqref{mU=NU unital q}, into the inequality
\begin{equation} \|L\|_2\,\leq\, 1\ . \label{hard proof eq1}\end{equation}
Take a generic filter $\psi=(N,c)$. In order to prove that $\mathcal{N}(\phi)=2$, we want to show that \mbox{$\phi\psi\phi=L(N,c)L=(LNL,Lc)\in\mathbf{EBt}_2$}. We will reach such a conclusion by applying~\eqref{TN cr}; indeed, we will show that
\begin{equation} \|LNL\|_1+|Lc|\, \leq\, 1\ . \label{hard proof eq2} \end{equation}

First of all, thanks to the case $p=1$, $r=s=2$ of~\eqref{Hoelder} and to~\eqref{hard proof eq1}, one has
\begin{equation} \|L N L \|_1\ \leq\ \|L N \|_2\, \|L \|_2\ \leq\ \|L N \|_2\ . \label{hard proof eq3}\end{equation}
Invoking Lemma~\ref{final}, and using again~\eqref{hard proof eq1}, we can see that
\begin{equation} \|LN\|_2\, +\, |Lc|\ \leq\ \|L\|_2\ \leq\ 1\ . \label{hard proof eq4} \end{equation}
Putting together~\eqref{hard proof eq3} and~\eqref{hard proof eq4}, we obtain exactly~\eqref{hard proof eq2}.
\end{proof}

Theorem~\ref{hard proof} strengthens the possibility that Conjecture~\ref{UFO q} could be true, by showing once again that \emph{there is nothing too similar to Example~\ref{cex UFO} in the $d=2$ case}. We stress however  that the problem of deciding the validity of Conjecture~\ref{UFO q} has been left open.

\section{Generalizations. An Example: LOCC Filters for Qubit Depolarizing Channels} \label{SEC:LOCC}

Through this paper, we mainly focused our attention on the simplest possible active action of Alice, i.e. the mere application of a local quantum channel. The main experimental advantage of this operation is that it can be realized in principle by using only natural processes and nor measurements neither communication. A generic $\mathbf{CPt}$ filter can be simply implemented by choosing a suitable coupling with a suitable environment and then discarding that ancillary system. We showed that even with these constraints a rich variety of cases arises.

However, other possible frameworks can be equally interesting. Here we want to discuss briefly the three main possible directions in which our scenario can be generalized.

\begin{itemize}

\item \emph{Type of operations}. In the main text we restrict ourselves to the local operations, performed separately by Alice and Bob in their own laboratories. This restriction can be removed by allowing more general LOCC or even separable operations (see~\cite{ENTAN,BENNETT}). Remind that in the LOCC case we allow Alice and Bob to communicate with a classical device, but to perform only local quantum operations. Instead, a separable operation is more generally a global channel on the bipartite system whose Kraus operators can be chosen as separated tensor products. From the above discussion, in ascending order of generality the three main cases are:
\begin{itemize}
\item local operations;
\item LOCC;
\item separable operations.
\end{itemize}  

\item \emph{Multistage coherence.} The basic scenario involves only a single--stage filtering. This constraint can be overcome if we allow Alice to perform generalized measurements on her subsystem, recording the corresponding outcomes. This information can be used in order to choose a more suitable operation on the next stage. If LOCC operations are allowed, she can even send the outcome to Bob, that can in turn use it to apply a clever filter, maybe in a later step.

Let us give an example of the latter strategy. Consider a measurement (performed on Alice's subsystem) which is described by the operators $\{ M_i \}_i$, and moreover a second set of measurements (again on $A$) labeled by the index $i$ and described by operators $\{ N^{(i)}_j \}_j$. For each pair $i,j$, let $ U_{ij} $ be a unitary matrix acting on $B$. For a generic bipartite input state $\rho$ and a local noisy channel $\phi$, the state
\begin{equation} \sum_{ij} (\phi\otimes I)\, (\mathcal{N}^{(i)}_j\otimes \mathcal{U}_{ij})\, (\phi\otimes I)\, (\mathcal{M}_i\otimes I)\, (\phi\otimes I)\, (\rho) \label{A eq0} \end{equation}
is the average state obtained through a multistage LOCC protocol. Here we employed the notation $\mathcal{M}_i(\cdot)=M_i(\cdot) M_i^\dag$ (and the same for $\mathcal{N}^{(i)}_j$ and $\mathcal{U}_{ij}$).

To summarize, the main distinction we can draw with respect to the internal coherence of the filtering process is between
\begin{itemize}
\item single--stage protocols; and
\item multistage protocols.
\end{itemize}

\item \emph{Definition of success.} When is the protocol successful in saving the entanglement? If only single--stage operations are allowed, there is no ambiguity. Instead, if generalized measurements are taken into account, there are at least three possible definitions.
\begin{itemize}
\item If for all the sequences of outcomes of the measurements the resulting bipartite state is entangled, we say that the entanglement has been saved \emph{deterministically}.
\item If there exists at least a sequence of outcomes that has nonzero probability of being realized and such that the corresponding state is entangled, we say that the entanglement has been saved \emph{probabilistically}.
\item If the average final state (obtained by forgetting the specific sequence of outcomes of the measurements) is entangled, we say that the entanglement has been saved \emph{on the average}.
\end{itemize}
Observe that both the deterministically successful protocols and the protocols which are successful on the average are necessarily also probabilistically successful. In this sense, the probabilistic framework is the most general. Instead, there is no a priori relation between deterministic and on--the--average protocols.

\end{itemize}

What about the initial entangled state Alice and Bob share? We will always suppose that they can optimize over it, choosing the most suitable for the preservation of the entanglement given a specific kind of noise. Observe that the optimization can be always restricted to the pure states, up to convex combinations. Since the noise is always local, if also the filtering operations are local (as in the main text) it is known that the optimal choice is always the maximally entangled state (a quantum channel is entanglement--breaking if and only if it breaks the entanglement of a maximally entangled state, see Subsection~\ref{subsec EB}).

Clearly, a scenario is nothing but a set of allowed protocols together with a definition of success. To each scenario a \emph{generalized entanglement--breaking index} $\mathcal{M}$ can be associated. Given a noise $\phi$, the integer $\mathcal{M}(\phi)$ is the smallest number of iteration of $\phi$ such that there is no filtering protocol in the chosen class which can successfully save the entanglement of any initial state.

The above discussion should convince the reader that a rich variety of (in principle) different situations can arise, depending on what filtering strategies we choose to allow. Clearly, the wider the class of protocols and the broader the definition of success are, the more difficult to compute the corresponding entanglement--breaking index is. In the remaining part of this section, we want to discuss the computation of the most general EB index (multistage separable operations with probabilistic success) of the simplest quantum channel (a depolarizing channel acting on a single qubit). 

We will introduce and discuss the various hypotheses one by one, because it is important to see where they come into play, determining a crucial simplification of the analysis. In the end, we will be able to perform the calculation we described.

Let us consider the most general scenario according to the above discussion. This means that we allow Alice and Bob to perform multistage separable operations on their bipartite system and require only a probabilistic success. In what follows, we will reserve capital Greek letters such as $\Phi$ or $\Psi$ to indicate global operations, while small letters will denote local channels, as usual.

As already discussed, a maximally entangled state does not need to be optimal for the entanglement preservation with respect to a fixed noise operation. For instance, we can not a priori exclude that there exists a local noise $\phi$ such that 
\begin{equation*} (\phi\otimes I)\, \Psi\, (\phi\otimes I)\, (\Ket{\varepsilon}\!\!\Bra{\varepsilon})\ \ \text{is separable $\ \forall$ separable $\Psi$,} \end{equation*}
but which admits a non--maximally entangled state $\Ket{\chi}\!\!\Bra{\chi}$ and a separable filter $\Psi_0$ such that
\begin{equation*} (\phi\otimes I)\, \Psi_0\, (\phi\otimes I)\, (\Ket{\chi}\!\!\Bra{\chi})\ \text{is entangled.} \end{equation*}
The same can happen for the LOCC operations instead of the separable ones. Of course this fact in principle complicates remarkably the analysis, even if we can always restrict the analysis to pure input states. This complication does not occur (in a sense, it is already present) in the case in which we require only probabilistic success, because every pure state can be obtained probabilistically from the maximally entangled state by including a suitable local measurement on Bob's subsystem at the first filtering stage.

Now, let us discuss why the specific choice of the depolarizing noise simplifies the analysis in the other non--probabilistic cases. Luckily enough, 
for the depolarizing noise $\Delta_\lambda$ (see~\eqref{D Ch})
the situation is still treatable, since cases similar to the one detailed above can be excluded. Indeed, for these maps starting with a maximally entangled state is always optimal.

\begin{thm} \label{D maxent opt} $\\$
Let $\rho$ be a bipartite state, and $-\frac{1}{d^2-1}\leq\lambda\leq 1$. Then there exists a LOCC operation $\Psi$ such that
\begin{equation} (\Delta_\lambda\otimes I)\, (\rho)\, =\, \left( \Psi\,(\Delta_\lambda\otimes I)\right)\, (\Ket{\varepsilon}\!\!\Bra{\varepsilon})\ . \label{D maxent opt eq} \end{equation}
\end{thm}

\begin{proof}
As above, up to convex combinations we can suppose that $\rho=\Ket{\chi}\!\!\Bra{\chi}$ is a pure state. Then, we must find a LOCC operation $\Psi$ such that
\begin{multline*} \lambda\,\Ket{\chi}\!\!\Bra{\chi}\, +\, (1-\lambda)\ \frac{\mathds{1}}{d}\otimes \text{Tr}_1 \Ket{\chi}\!\!\Bra{\chi}\ =\\
=\ (\Delta_\lambda\otimes I) (\Ket{\chi}\!\!\Bra{\chi})\ =\ \Psi\,(\Delta_\lambda\otimes I)\, (\Ket{\varepsilon}\!\!\Bra{\varepsilon})\ =\\
=\ \lambda\ \Psi(\Ket{\varepsilon}\!\!\Bra{\varepsilon})\ +\ (1-\lambda)\ \Psi\left(\frac{\mathds{1}}{d^2}\right)\ , \end{multline*}
which can be fulfilled for instance by having
\begin{equation} \Psi(\Ket{\varepsilon}\!\!\Bra{\varepsilon})\, =\, \Ket{\chi}\!\!\Bra{\chi}\ \ \text{and}\ \ \Psi\left(\frac{\mathds{1}}{d^2}\right)\, =\, \frac{\mathds{1}}{d}\otimes \text{Tr}_1 \Ket{\chi}\!\!\Bra{\chi}\, . \label{D maxent opt request} \end{equation}


In Ref.~\cite{NielsenLOCC} it is proved that there exists a LOCC operation transforming $\Ket{\alpha}$ to $\Ket{\beta}$ if and only if the spectrum of $\text{Tr}_1 \Ket{\alpha}\!\!\Bra{\alpha}$ majorizes that of $\text{Tr}_1 \Ket{\beta}\!\!\Bra{\beta}$. With the notation of~\eqref{eq majorization}, we write this condition as
\begin{equation} \text{Tr}_1 \Ket{\beta}\!\!\Bra{\beta}\ \prec\ \text{Tr}_1 \Ket{\alpha}\!\!\Bra{\alpha}\ . \label{D maxent opt maj} \end{equation}
Naturally, this means that the maximally entangled state can be transformed to whatever pure state. As a consequence, we can find a LOCC operation $\Psi$ which satisfies the first condition of~\eqref{D maxent opt request}. Moreover such a LOCC protocol can be composed of only two steps: firstly, a measurement on Bob's subsystem; and secondly, an unitary transformation on Alice's subsystem. As a consequence, the whole trasformation can be written as
\begin{equation} \Psi(\cdot)\, =\, \sum_i\, U_i\otimes M_i\,(\cdot)\, U_i^\dag\otimes M_i^\dag\ , \label{D maxent opt LOCC} \end{equation}
where the $U_i$ are unitary matrices, and $\sum_i M_i^\dag M_i\, =\, \mathds{1}$. By virtue of~\eqref{D maxent opt LOCC}, one easily obtains
\begin{equation*} \text{Tr}_1 \Ket{\chi}\!\!\Bra{\chi}\, =\, \text{Tr}_1 \Psi(\Ket{\varepsilon}\!\!\Bra{\varepsilon})\, =\, \frac{1}{d}\, \sum_i\, M_i M_i^\dag\ . \end{equation*}
But then
\begin{equation*} \Psi\left(\frac{\mathds{1}}{d^2}\right)\, =\, \frac{\mathds{1}}{d}\otimes\,\frac{1}{d}\, \sum_i M_i M_i^\dag\, =\, \frac{\mathds{1}}{d}\otimes \text{Tr}_1 \Ket{\chi}\!\!\Bra{\chi}\ , \end{equation*}
and so also the second condition of~\eqref{D maxent opt request} is met.
\end{proof}

Theorem~\ref{D maxent opt} shows that using a maximally entangled state is always optimal for a depolarizing noise. Indeed, suppose for instance that
\begin{equation*} (\Delta_\lambda\otimes I)\, \Phi\, (\Delta_\lambda\otimes I)\, (\Ket{\varepsilon}\!\!\Bra{\varepsilon})\ \ \text{is separable $\ \forall\ \Phi\in\mathbf{LOCC}$.} \end{equation*}
Take a generic input state $\Ket{\chi}\!\!\Bra{\chi}$; then for all LOCC filters $\Phi$ we find that
\begin{equation*} (\Delta_\lambda\otimes I)\, \Phi\, (\Delta_\lambda\otimes I)\, (\Ket{\chi}\!\!\Bra{\chi})\ =\ (\Delta_\lambda\otimes I)\, \Phi\Psi\, (\Delta_\lambda\otimes I)\, (\Ket{\varepsilon}\!\!\Bra{\varepsilon}) \end{equation*}
is a separable state (this is a trivial consequence of the fact that $\Phi\Psi$ is still LOCC). A similar reasoning of course holds for the separable class of filters.

Now, let us examine in greater detail the single qubit case $d=2$. The importance of this restriction will be clear soon. The main result of this section is the proof that \emph{even multistage separable filtering with probabilistic success is completely useless when a depolarizing noise is acting locally on a single qubit}. This means that the corresponding generalized EB index, indicated with $\mathcal{M}_{SEP}^{ms,\, pr}$, takes the same value as the direct $n$--index:

\begin{equation} \mathcal{M}_{SEP}^{ms,\, pr}(\Delta_\lambda)\ =\ \Big\lceil\ \frac{\log 3}{\,\log\,\frac{1}{\lambda}\,}\ \Big\rceil \label{M D} \end{equation}

Interestingly enough,~\eqref{M D} is a significant improvement of~\eqref{N D}, where only local, single--stage filtering is considered. Of course, we left open the problem for higher dimension, and it could be of some interest finding either an example of an improvement occurring when more general filtering strategies are allowed, or on the contrary a general proof that $\mathcal{M}_{SEP}^{ms,\,pr}(\Delta_\lambda)=\,n\,(\Delta_\lambda)$.

Firstly, we want to fix some notation. Following~\cite{EA1,EA2}, we will indicate with $\mathbf{EA}$ ($\mathbf{EAt}$) the convex set of \emph{entanglement--annihilating} (and trace--preserving, respectively) maps acting on a bipartite system, that is, the set of maps that always produce a separable output regardless of what input state they are acting on. Since in what follows our bipartite system will be made of two qubits, it will be not necessary to specify the dimension with an appropriate subscript. Another useful convex set of channels acting on a bipartite system is the already discussed class of separable (and trace--preserving) maps, denoted by $\mathbf{S}$ ($\mathbf{St}$, respectively). By definition, these channels always preserve the separability of states. The LOCC protocols (which again constitute a convex set) are special examples of separable channels; formally, $\mathbf{LOCC}\subset\mathbf{St}$ (and it is known that the inclusion is strict). Some basic properties are summarized as follows.
\begin{gather} \phi\in\mathbf{EBt}\quad\Rightarrow\quad\phi\otimes I\in\mathbf{EAt} \label{A eq1} \\
\Phi\in\mathbf{S},\ \Psi\in\mathbf{EA}\quad \Rightarrow\quad \Phi\Psi\in\mathbf{EA} \label{A eq2} \\
\Phi\in\mathbf{P},\ \Psi\in\mathbf{EA}\quad \Rightarrow\quad \Psi\Phi\in\mathbf{EA} \label{A eq3}
\end{gather}
In~\eqref{A eq3}, the symbol $\mathbf{P}$ represents the set of positive maps acting on a two--qubit system.

Next, let us introduce some mathematical tool in order to take advantage of our restriction to the qubit case. In~\cite{HorodeckiDep} the reduction map
\begin{equation} P\, =\, \mathds{1}\otimes\text{Tr}_1 \, -\, I \label{R ch} \end{equation}
is introduced and studied. Let us recap some of the most important results.
\begin{itemize}
\item $P$ is trace--preserving, thanks to the restriction to the $d=2$ case.
\item Although \emph{not} positive in general (on the global bipartite system), $P$ produces a well--behaved separable state when acting on a separable state. Formally, we write
\begin{equation} \rho\in\mathcal{S}\quad\Rightarrow\quad 0\leq P(\rho)\in\mathcal{S}\ , \label{A eq4} \end{equation}
$\mathcal{S}$ being the set of separable states. As a consequence, we obtain
\begin{equation} \Psi\in\mathbf{EAt}\quad\Rightarrow\quad P\Psi\in\mathbf{EAt}\ . \label{A eq5} \end{equation}

\item Since the equality $P\,=\,TV_1\otimes I$ holds (see~\eqref{V Ch}), it follows from~\eqref{V comp} that $P$ commutes with the depolarizing channels. Moreover,
\begin{equation} P^2 = I\, . \label{Psquared} \end{equation}
\end{itemize}

Now, we are ready to state the main theorem of this appendix. Although it could be regarded as purely technical, this tool will turn out to be fundamental in order to prove our claim.

\begin{thm} \label{LOCC D} $\\$
Let $\Psi_1,\ldots,\Psi_{n-1}$ be linear maps acting on a two--qubit system, and take $0\leq\lambda_1,\ldots,\lambda_n\leq 1$. Then the map
\begin{equation} (\Delta_{\lambda_1}\otimes I)\,\Psi_1\,(\Delta_{\lambda_2}\otimes I)\ \ldots\ (\Delta_{\lambda_{n-1}}\otimes I)\,\Psi_{n-1}\,(\Delta_{\lambda_n}\otimes I) \label{LOCC D eq1} \end{equation}
can be written as a convex combination of terms of the form
\begin{equation} \tilde{\Psi}_1\ldots\tilde{\Psi}_i\,(\Delta_{\lambda_1\ldots\lambda_n}\otimes I)\,\Psi_{i+1}\ldots\Psi_{n-1}\ , \label{LOCC D eq2} \end{equation}
where $0\leq i\leq n-1$ (the two extreme values corresponding to the degenerate terms \mbox{$(\Delta_{\lambda_1\ldots\lambda_n}\otimes I)\,\Psi_1\ldots\Psi_{n-1}$} and \mbox{$\tilde{\Psi}_1\ldots\tilde{\Psi}_{n-1}\,(\Delta_{\lambda_1\ldots\lambda_n}\otimes I)$}), and the generic symbol $\tilde{\Psi}_j$ represents either the map $\Psi_j$ or the map $P\,\Psi_j P$.
\end{thm}

\begin{proof}
Let us prove the theorem by induction. The first nontrivial case $n=2$ can be solved directly.
\begin{multline} (\Delta_\lambda\otimes I)\,\Psi\,(\Delta_\mu\otimes I)\ \ =\ \ \frac{(1+\lambda)(1-\mu)}{2(1-\lambda\mu)}\ \Psi\,(\Delta_{\lambda\mu}\otimes I)\ \ +\\
+\ \ \frac{(1-\lambda)(1-\mu)}{2(1+\lambda\mu)}\ P\,\Psi P\,(\Delta_{\lambda\mu}\otimes I)\ \ +\\
+\ \ \frac{\mu(1-\lambda^2)}{1-\lambda^2 \mu^2}\ (\Delta_{\lambda\mu}\otimes I)\, \Psi\ \ . \label{LOCC D eq3} \end{multline}
Now, suppose that we proved the thesis for $n$, and let us examine the $n+1$ case. By applying the inductive hypothesis, we can decompose the map
\begin{equation*} (\Delta_{\lambda_1}\otimes I)\,\Psi_1\,(\Delta_{\lambda_2}\otimes I)\ \ldots\ (\Delta_{\lambda_n}\otimes I)\,\Psi_n\,(\Delta_{\lambda_{n+1}}\otimes I) \end{equation*}
as a convex combination of terms of the form
\begin{equation} \tilde{\Psi}_1\ldots\tilde{\Psi}_i\,(\Delta_{\lambda_1\ldots\lambda_n}\otimes I)\,\Psi_{i+1}\ldots\Psi_{n-1}\,\Psi_n\,(\Delta_{\lambda_{n+1}}\otimes I) \label{LOCC D eq4} \end{equation}
(for the notation see the above explanation). Now, equation~\eqref{LOCC D eq3} allows us to write
\begin{equation*} (\Delta_{\lambda_1\ldots\lambda_n}\otimes I)\,\Psi_{i+1}\ldots\Psi_{n-1}\,\Psi_n\,(\Delta_{\lambda_{n+1}}\otimes I) \end{equation*}
as a convex combination of the three terms
\begin{gather*} \Psi_{i+1}\ldots\Psi_{n-1}\,\Psi_n\,(\Delta_{\lambda_1\ldots\lambda_n\lambda_{n+1}}\otimes I)\ ,\\
P\,\Psi_{i+1}\ldots\Psi_{n-1}\,\Psi_n P\,(\Delta_{\lambda_1\ldots\lambda_n\lambda_{n+1}}\otimes I)\ ,\\
(\Delta_{\lambda_1\ldots\lambda_n\lambda_{n+1}}\otimes I)\, \Psi_{i+1}\ldots\Psi_{n-1}\,\Psi_n\ . \end{gather*}
Once inserted into~\eqref{LOCC D eq4}, all these terms fit into the general form~\eqref{LOCC D eq2} (taking into account also also~\eqref{Psquared}), and we are done.
\end{proof}

As a corollary of this result, we can easily prove that multistage separable filters can not save the entanglement against a local depolarizing noise, even if only probabilistic success is required. This is the same as to claim that~\eqref{M D} holds.
Suppose that a general, multistage separable filter acts on an input bipartite state $\rho$, producing a specific sequence of outcomes of the measurements. Then, it should be clear that the final state of the system is proportional to
\begin{equation*} (\Delta_{\lambda_1}\otimes I)\,\Psi_1\,(\Delta_{\lambda_2}\otimes I)\ \ldots\ (\Delta_{\lambda_{n-1}}\otimes I)\,\Psi_{n-1}\,(\Delta_{\lambda_n}\otimes I)\ (\rho)\ , \end{equation*}
where the $\Psi_i$ are separable, in general non--trace--preserving maps (see also~\eqref{A eq0}). If we could prove that this (unnormalized) state is always separable, then the thesis~\eqref{M D} would follow. All that is implied by the following result.\\

\begin{cor} \label{LOCC D final} $\\$
Let $\Psi_1,\ldots,\Psi_{n-1}\in\mathbf{S}$ be separable maps acting on a two--qubit bipartite system, and take \mbox{$-\frac{1}{3}\leq\lambda_1,\ldots,\lambda_n\leq 1$} such that $\lambda_1\ldots\lambda_n\leq \frac{1}{3}$. Then the map
\begin{equation} (\Delta_{\lambda_1}\otimes I)\,\Psi_1\,(\Delta_{\lambda_2}\otimes I)\ \ldots\ (\Delta_{\lambda_{n-1}}\otimes I)\,\Psi_{n-1}\,(\Delta_{\lambda_n}\otimes I) \label{LOCC D final eq} \end{equation}
is entanglement--annihilating.
\end{cor}

\begin{proof}
If $\lambda_i\leq 0$ for some $i$, then the channel $\Delta_{\lambda_i}$ is entanglement--breaking and the thesis follows from~\eqref{A eq1}. Otherwise, Theorem~\ref{LOCC D} shows that the map~\eqref{LOCC D final eq} can be written as a convex combination of terms of the form
\begin{equation*} \tilde{\Psi}_1\ldots\tilde{\Psi}_i\,(\Delta_{\lambda_1\ldots\lambda_n}\otimes\, I)\,\Psi_{i+1}\ldots\Psi_{n-1}\ , \end{equation*}
where we use the same notation of~\eqref{LOCC D eq2}. Now, we want to prove that every such a term is indeed entanglement--annihilating. The thesis will follow thanks to the convexity of the $\mathbf{EA}$ set.
Since $\lambda_1\ldots\lambda_n\leq \frac{1}{3}$,~\eqref{D EB} guarantees that $\Delta_{\lambda_1\ldots\lambda_n}\in\mathbf{EBt}$. As a consequence,~\eqref{A eq1} implies that $\Delta_{\lambda_1\ldots\lambda_n}\otimes\, I\in\mathbf{EAt}$. By applying~\eqref{A eq3}, it follows that also
\begin{equation*} (\Delta_{\lambda_1\ldots\lambda_n}\otimes\, I)\,\Psi_{i+1}\ldots\Psi_{n-1}\, \in\, \mathbf{EA}\, . \end{equation*}
Finally, using~\eqref{A eq2} and~\eqref{A eq5} we can conclude that
\begin{equation*} \tilde{\Psi}_1\ldots\tilde{\Psi}_i\,(\Delta_{\lambda_1\ldots\lambda_n}\otimes\, I)\,\Psi_{i+1}\ldots\Psi_{n-1}\ \in\ \mathbf{EA}\ . \end{equation*}
\end{proof}

\section{Conclusions}\label{sec conc}
In this paper we introduced and studied some functionals to characterize the entanglement--breaking behavior of the iterations of a given channel. The scenario in which Alice is allowed to play an active role against the noise, by interposing appropriate local filters between two consecutive applications of the noise, is deeply analyzed. Also the general case in which Alice and Bob can implement a wider class of filtering protocols is presented, and the most general corresponding index is computed for a depolarizing noise acting on a two--qubit system. In the basic case, we examined some examples in which the best strategy can be provably found, and all our indices analytically calculated. Furthermore, a nontrival counterexample (valid in all dimensions $d\geq 3$) is presented, showing that the best filtering strategy is \emph{not} always unitary. We remark that the corresponding question for the qubit case is still open, even if we were able to prove two partial results pointing out that such a counterexample could not exist for two--dimensional systems. Another interesting problem on which we are currently working, is the characterization of those channels exhibiting divergent values of some entanglement--breaking indices.

We thank A. De Pasquale for comments and discussions.


\begin{thebibliography}{99}

\bibitem{HOLEVOBOOK} A. S. Holevo, \emph{Quantum Systems, Channels, Information} (de Gruyter Studies in Mathematical Physics, 2012).

\bibitem{WolfQC} M. M. Wolf, \emph{Quantum Channels \& Operations}, Lecture Notes available at http://www-m5.ma.tum.de/foswiki/pub/M5/Allgemeines\\/MichaelWolf/QChannelLecture.pdf (2012).

\bibitem{BennettShor} C. H. Bennett,  P. W. Shor,  IEEE Trans. Inf. Theory {\bfseries 44}(6), 2724-2742 (1998).

\bibitem{ENTAN} R. Horodecki, P. Horodecki, M. Horodecki, and K. Horodecki, Rev. Mod. Phys. {\bf 81}, 865 (2009).

\bibitem{V} A. De Pasquale, V. Giovannetti, Phys. Rev. A {\bfseries 86}, 052302 (2012).

\bibitem{MULLER}A. M\"{u}ller-Hermes, D. Reeb, and M. M. Wolf, IEEE Trans. Inf. Theory {\bfseries 61}(1), 565-581 (2015).

\bibitem{Vnonmaxent} A. De Pasquale, A. Mari, A. Porzio, V. Giovannetti, Phys. Rev. A {\bfseries 87}, 062307 (2013).

\bibitem{BENNETT} C. H. Bennett, D. P. DiVincenzo, C. A. Fuchs, T. Mor, E. Rains, P. W. Shor, J. A. Smolin, and W. K. Wootters, Phys. Rev. A {\bf 59}, 1070 (1999).

\bibitem{HorodeckiShorRuskai} M. Horodecki, P. W. Shor, M. B. Ruskai, Rev. Math. Phys. {\bfseries 15}, 629 (2003).

\bibitem{Holevoform} A. S. Holevo, Russian Math. Surveys {\bfseries 53}, 1295-1331 (1999).

\bibitem{AlgoetFujiwara} A. Fujiwara, P. Algoet, Phys. Rev. A {\bfseries 59}, 3290 (1999).

\bibitem{RuskaiEBqubit} M. B. Ruskai, Rev. Math. Phys. {\bfseries 15}, 643 (2003).

\bibitem{NC} M. A. Nielsen, I. L. Chuang, \emph{Quantum Computation and Quantum Information} (Cambridge University Press, Cambridge, 2000).

\bibitem{PeresPPT} A. Peres, Phys. Rev. Lett. {\bfseries 77}, 1413 (1996).

\bibitem{HorodeckiPPT} M. Horodecki, P. Horodecki, R. Horodecki, Phys. Lett. A {\bfseries 223}, 1 (1996).

\bibitem{HorodeckiDep} M. Horodecki, P. Horodecki, Phys. Rev. A {\bfseries 59}, 4206 (1999).

\bibitem{Werner} R. F. Werner, Phys. Rev. A {\bfseries 40}, 4277 (1989).

\bibitem{PATA} F. Pastawski, L. Clemente, and J. I. Cirac,  Phys. Rev. A, {\bf 83} 012304 (2011).

\bibitem{KingRuskai} C. King, M. B. Ruskai, IEEE Trans. Info. Theory {\bfseries 47}, 192-209 (2001).

\bibitem{RuskaiCPqubit} M. B. Ruskai, S. Szarek, W. Werner, Lin. Alg. Appl. {\bfseries 347}, 159 (2002).

\bibitem{GeometryQuantum} I. Bengtsson and K. $\mathrm{\dot{Z}yczkowski}$, \emph{Geometry of Quantum States} (Cambridge University Press, Cambridge, 2006).

\bibitem{J} J. I de Vicente, J. Phys. A: Math. Theor. {\bfseries 41}, 065309 (2008).

\bibitem{HJ1} R. A. Horn and C. R. Johnson, \emph{Matrix Analysis} (Cambridge University Press, Cambridge, 1990).

\bibitem{NielsenLOCC} M. A. Nielsen, Phys. Rev. Lett. {\bfseries 83}, 436-439 (1999).

\bibitem{EA1} L. Morav$\check{\text{c}}$\'ikov\'a and M. Ziman, J. Phys. A: Math. Theor. {\bfseries 43}, 275306 (2010).

\bibitem{EA2} S. N. Filippov, T. Ryb\'ar and M. Ziman, Phys. Rev. A {\bfseries 85}, 012303 (2012).



















\end{thebibliography}
\end{document}